\documentclass[3p]{elsarticle}
\usepackage{amsmath}
\usepackage{amssymb}
\usepackage{bm}
\usepackage[mathscr]{eucal}
\usepackage{theorem}
\usepackage{graphicx}
\usepackage{psfrag}
\usepackage{subfigure}
\usepackage{color}
\usepackage{wasysym}
\include{graphix}
\usepackage{framed}
\usepackage{enumitem}%
\usepackage{lipsum}
\usepackage{float}
\usepackage{MnSymbol}
\usepackage{url}

\newcommand{\mathe}{\mathrm{e}}
\newcommand{\mathd}{\mathrm{d}}

\newcommand{\vecx}{\bm{x}}

\newcommand{\scalegamma}{g}
\newcommand{\betaconst}{\beta_0}

\graphicspath{{figs/}}

\newtheorem{theorem}{Theorem}[section]
\newtheorem{proposition}[theorem]{Proposition}

\newenvironment{proof}[1][Proof]{\begin{trivlist}
\item[\hskip \labelsep {\bfseries #1}]}{\end{trivlist}}

\newcommand{\myqed}{\nobreak \ifvmode \relax \else
      \ifdim\lastskip<1.5em \hskip-\lastskip
      \hskip1.5em plus0em minus0.5em \fi \nobreak
      \vrule height0.75em width0.5em depth0.25em\fi}

\begin{document}
\title{A large-scale statistical study of the coarsening rate in models of Ostwald-Ripening}
\author[mymainaddress]{Lennon \'O N\'araigh\corref{mycorrespondingauthor}}
\cortext[mycorrespondingauthor]{Corresponding author}
\ead{onaraigh@maths.ucd.ie}

\author[mymainaddress]{Andrew Gloster}

\address[mymainaddress]{School of Mathematics and Statistics, University  College Dublin, Belfield, Dublin 4, Ireland}

\date{\today}

\begin{abstract}
In this article we look at the coarsening rate in two standard models of Ostwald Ripening.  Specifically, we look at a discrete droplet popoulation model, which in the limit of an infinite droplet population reduces to the classical Lifshitz--Slyozov--Wagner model.  We also look at the Cahn--Hilliard equation with constant mobility.
We define the coarsening rate as $\beta=-(t/F)(\mathd F/\mathd t)$, where $F$ is the total free energy of the system and $t$ is time.  There is a conjecture that the long-time average value of $\beta$ should not exceed $1/3$ -- this result is summarized here as $\langle \beta\rangle\leq 1/3$.  We explore this conjecture for the two considered models.  Using large-scale computational resources (speficially, GPU computing employing thousands of threads), we are able to construct ensembles of simulations and thereby build up a statistical picture of $\beta$.  Our results show that the droplet population model and the Cahn--Hilliard equation (asymmetric mixtures) are demonstrably in agreement with $\langle\beta\rangle\leq 1/3$.  The results for the Cahn--Hilliard equation in the case of symmetric mixtures show $\langle\beta\rangle$ sometimes exceeds $1/3$ in our simulations.  However, the possibility is left open for the very long-time average values of $\langle \beta\rangle$ to be bounded above by $1/3$.  The theoretical methodology laid out in this paper sets a path for future more invensive computational studies whereby this conjecture can be explored in more depth.
\end{abstract}

\maketitle

\section{Introduction}
\label{sec:intro}

In a seminal paper, Cahn and Hilliard~\cite{CH_orig} introduced their eponymous equation to model the dynamics of phase separation. They pictured a binary alloy in a mixed state, cooling below a critical temperature. This state is unstable to small perturbations so that fluctuations cause the alloy to separate into domains rich in one material or the other.  The domains grow in time in a phenomenon called coarsening. Because the evolution of the concentration field is an order-parameter equation, the Cahn--Hilliard equation gives a completely general description of phase separation, applicable not only to molten alloys but also to polymer mixtures~\cite{cabral2018spinodal}, two-phase flows~\cite{ding2007diffuse}, and nanobubbles~\cite{tomo2018unexpected}.  
Over the years, numerical simulations~\cite{spectralCahn} and scaling arguments~\cite{LennonPRE2007} have been used to establish that the typical domain size $L(t)$ expands in time as $L(t)\sim t^{1/3}$.  Kohn and Otto~\cite{kohn2002upper} have established the following rigorous upper bound on the coarsening rate in the Cahn--Hilliard equation:
\begin{equation}
\frac{1}{T}\int_0^T L^{-2}(t)\mathd t\geq \frac{K}{T}\int_0^T \left(t^{-1/3}\right)^2\mathd t,
\label{eq:kohn}
\end{equation}
where $K$ is constant.  The constant $K$ depends on the dimensionality of the space and on $\Omega$, the region occupied by the binary fluid; we take this opportunity also to introduce $|\Omega|$, the volume of the region $\Omega$.    In this context, $L(t)$ is computed as $L(t)=|\Omega|/\left[\text{Total interfacial area at time }t\right]$.  Although the bound~\eqref{eq:kohn} has been proved rigorously, it is not known if there is an instantaneous (pointwise) analogue, that is, a bound of the form 
\begin{equation}
L(t)\stackrel{\text{?}}{\leq} (\text{Const.}) t^{1/3},
\label{eq:kohn_maybe}
\end{equation}
valid for sufficiently late time $t$, where the putative constant is positive and depends only on $\Omega$ and $D$.  The aim of this article is to carry out numerical simulations to shed light on this problem: our ensembles of simulations hint at the existence of such a bound.

Before presenting the results of the simulations we review the mathematical theory of the Cahn--Hilliard equation.  This is presented here as:
\begin{equation}
\frac{\partial C}{\partial t}=\nabla^2(C^3-C-\gamma\nabla^2C),\qquad \vecx\in \Omega,\qquad t>0,
\label{eq:ch}
\end{equation}
where $\gamma$ is a small positive parameter, and $C$ is a volume fraction tracking the abundance of the different binary fluid components, with $C=\pm 1$ corresponding to the pure phases.  
Under suitable boundary conditions on $\partial\Omega$, the Cahn--Hilliard equation~\eqref{eq:ch} reduces the following free energy:
\begin{equation}
F=\int_{\Omega}\left[\tfrac{1}{2}\left(C^2-1\right)^2+\tfrac{1}{2}\gamma|\nabla C|^2\right]\mathd^D x,\qquad
\frac{\mathd F}{\mathd t}=-\int_{\Omega}\left|\nabla (C^3-C-\gamma\nabla^2C)\right|^2\mathd^D x.
\label{eq:dFdt_ch}
\end{equation}
Here, $D$ is the dimension of the space, which in our investigations, will be set equal to either $D=2$ or $D=3$, as required.
Solutions of Equation~\eqref{eq:ch} are characterized by a rapid relaxation to $C=\pm 1$ locally, in domains, followed by slow domain growth -- this evolution is driven by the energy-minimization~\eqref{eq:dFdt_ch} and the conservation law $(\mathd/\mathd t)\int_\Omega C\,\mathd^D x=0$, the latter being a further consequence of the structure of Equation~\eqref{eq:ch} and the assumed boundary conditions.

For the case of a minority phase immersed in a matrix of the majority phase, small `droplets' of the minority phase in the Cahn--Hilliard equation disappear only to be reabsorbed into larger droplets of the same phase in precisely the same process as Ostwald Ripening~\cite{baldan2002review}.    
There is a simple asymptotic theory for Ostwald Ripening, valid in an infinitely large domain with infinitely many droplets, but formulated in such a way that the volume fraction occupied by the droplets is finite (and small).  
The theory was developed by Lifshitz and Slyzov in Reference~\cite{lifshitz1961kinetics} and simultaneously, by Wagner in Reference~\cite{wagner1961} -- it is therefore referred to as the LSW theory.
 The LSW is independent of the Cahn--Hilliard equation; however, it can be shown~\cite{pego1989} that solutions of the Cahn--Hilliard equation~\eqref{eq:ch} (under certain assumptions) coincide with solutions of the LSW theory, in the limit as $\gamma\rightarrow 0$.   Because of its usefulness in making sense of  numerical simulations of the Cahn--Hilliard equation in the dilute limit, the LSW theory is reviewed here.

The main result of the LSW theory is an analytical expression for the late-time dropsize distribution in Ostwald Ripening, denoted here by $p(r,t)$.  
A continuity equation for $p(r,t)$ is derived, where the probability flux depends on the droplet velocity -- the velocity is obtained by energetic arguments based on droplet interfacial area.  
The continuity equation admits a self-similar solution $p(r,t)\propto f(x)$, where $x=r/\langle R\rangle$, and where $\langle R\rangle$ is the instantaneous value of the mean droplet radius; this gives an alternative characterization of the typical domain size.  The scaling behaviour $\langle R\rangle\propto t^{1/3}$ readily drops out of this calculation.  
It can be also noted that the functional form of $f(x)$  has compact support.  

The LSW theory represents an approximate solution to so-called Mullins--Sekerka (MS) Dynamics, introduced first by Mullins and Sekerka to model particle growth in a supersaturated matrix~\cite{mullins1963morphological}, but then re--purposed as an effective model for Ostwald Ripening more generally~\cite{niethammer2008effective}.  
The MS dynamics describe the motion of extended regions $\{B_1,\cdots,B_N\}$ in a domain $\Omega$ and are expressible in terms of a generic chemical potential $\mu$:
\begin{subequations}
\begin{eqnarray}
\nabla^2\mu&=&0,\qquad \text{in }\Omega - \cup_{i=1}^N B_i,\\
\mu&=&\kappa,\qquad \text{in } \cup_{i=1}^N \partial B_i,\\
V&=&\left[\widehat{\bm{n}}\cdot\nabla \mu\right],\qquad \text{in } \cup_{i=1}^N \partial B_i.\label{eq:ms3}
\end{eqnarray}%
\label{eq:ms}%
\end{subequations}%
Here, $\kappa$ denotes the mean interfacial curvature, $V$ denotes the normal velocity of the interface, $\widehat{\bm{n}}$ denotes the normal vector to the interface, and $\left[\widehat{\bm{n}}\cdot\nabla \mu\right]$ denotes the jump in the normal derivative of the chemical potential across the interface.  
The fact that the interfaces move (with velocity $V$, via a mismatch in the chemical potential across the interfaces), means that this is a dynamical problem.  
LSW theory amounts to a solution of Equation~\eqref{eq:ms} in the mean-field approximation, for spherical domains $B_i$.

Motivated by these discussions, in the present work we perform numerical simulations for two distinct but related models:

\noindent{\textbf{Model 1:}}  Droplet populations, where the dynamics are governed by a set of discrete equations for the droplet radii.  The initial number of droplets $N$ is taken to be large but finite. 

\noindent{\textbf{Model 2:}}  Droplet populations, where the dynamics are governed by the Cahn--Hilliard equation~\eqref{eq:ch} for a single scalar field $C(\bm{x},t)$.  The size of the region $\Omega$ in which $C$ is defined is taken to be large but finite.

\noindent In both cases, we are motivated to consider the total interfacial area $F$ as a measure of surface energy; this then gives $L(t)=|\Omega|/F$ as the typical lengthscale.  In the case where the dynamics are governed by the Cahn--Hilliard equation, $F$ can be identified with the free energy in Equation~\eqref{eq:dFdt_ch}.  We then  identify the {\textit{coarsening rate}} $\beta$:
\begin{equation}
\beta=\frac{t}{L}\frac{\mathd L}{\mathd t}=-\frac{t}{F}\frac{\mathd F}{\mathd t}
\label{eq:betadef}
\end{equation}
for both cases.  Note that $\beta$ is a property of the system, i.e. of the entire droplet population.  We look at the probability distribution function  of $\beta$; $\beta$ is viewed as a probabilistic variable
that emerges from performing an ensemble of different numerical simulations:
\begin{equation}
p_\beta(b,t)=\left[\begin{array}{c}
\text{Probability that a given simulation}\\
\text{produces a growth rate }\beta\text{ in the range }\\
b\leq \beta\leq b+\mathd b,\text{ at time }t\end{array}\right].
\label{eq:p_beta}
\end{equation}
In Case 1 the ensemble is made up of $M$ simulations, each with $N$ interacting particles present initially; all initial conditions for the particle radii are random numbers drawn from the uniform distribution.
In Case 2 the ensemble is made up of $M$ different simulations of Equation~\eqref{eq:ch}, where again, each initial condition has a random initial condition made up (in an appropriate sense) from the uniform distribution.
The structure of the resulting probability distribution functions will give some clue if $L(t)\leq (\text{Const.}) t^{1/3}$ holds pointwise, or only in an averaged sense, as in Equation~\eqref{eq:kohn}.  More concretely, the plan of the paper is to answer the following questions:
%

\noindent{\textbf{Question 1:}}  Is it sensible even to define a probability distribution function for $\beta$ in Equation~\eqref{eq:betadef}? 
This question is answered in the affirmative by reference to LSW theory; this corresponds to Model 1 with $N\rightarrow\infty$ . In this instance, the distribution of $\beta$ can be computed analytically.  These results are established in Section~\ref{sec:LSW}.  This then justifies the formulation of analogous probability distribution functions for Model 1 and Model 2.

\noindent{\textbf{Question 2:}}   What is  the probability distribution function for $\beta$ in Model 1?  How do finite-size effects (parametrized by $N$, the number of droplets initially present) alter the shape of the distribution?  These questions are answered in Section~\ref{sec:drop_pop}, where we compute $p_{\beta}^N(b,t)$ via numerical simulation -- i.e. the probability distribution function in Equation~\eqref{eq:p_beta}, with the finite-size effect accounted for.

\noindent{\textbf{Question 3:}}  What is the probability distribution function for $\beta$ in Model 2?  How do finite-size effects (parametrized by $|\Omega|$) alter the shape of the distribution?  Can the distributions for Model 1 and 2 be compared? These questions are answered in Section~\ref{sec:CH}, where we introduce the notation  $p_{\beta}^\Omega(b,t)$ for  the probability distribution function in Model 2, with finite-size effects properly accounted for.  In particular, $p_{\beta}^\Omega(b,t)$ is built up from an ensemble consisting of $M$ numerical simulations of Equation~\eqref{eq:ch}, each with random initial conditions. The prospect of generating such ensembles is made computationally feasible using Graphics Processing Units (GPUs), which we describe in detail. 
\vspace{0.1in}

\noindent By answering these three questions, we demonstrate numerically that the probability distribution functions 
$p_{N}^\Omega(b,t)$ and $p_{\beta}^\Omega(b,t)$ are not self-similar -- the moments of the probability distribution functions vary systematically over time.  As such, we are led to consider a stochastic model for $\beta$,
\begin{equation}
\beta=\betaconst+ \delta\beta,
\label{eq:beta_model}
\end{equation}
where $\betaconst$ is a constant, and where $\delta\beta$ is a piecewise-continuous function of time with jump discontinuities occurring at random times as the system evolves.
Hence, by integrating Equation~\eqref{eq:betadef} with respect to time, from $t_0$ to $t$, we are led to:
\begin{equation}
F(t)=F(t_0)(t/t_0)^{-\betaconst}\mathe^{-\int_{t_0}^t (\delta\beta/t)\mathd t},
\label{eq:beta_model1}
\end{equation}
Since $\delta\beta$ is piecewise-continuous (with jumps at well-spaced random intervals), the integral in Equation~\eqref{eq:beta_model1} can be interpreted as an ordinary Riemann integral.
Taking expectation values with respect to the measure induced by the random jumps in $\delta\beta$, we obtain:
\begin{equation}
\mathbb{E}\bigg\{\left[\log\left(\frac{F(t_0)t_0^{\betaconst}}{F(t)t^{\betaconst}}\right)\right]^2\bigg\} \leq (t-t_0)\int_{t_0}^{t}\frac{\mathbb{E}(\delta\beta^2)}{t^2}\mathd t,\qquad t>t_0.
\label{eq:beta_model2}
\end{equation}
Our simulation results indicate that $\mathbb{E}(\delta\beta^2)=k t$, where $k$ is a constant,  hence we are led to propose the following bound (in the mean-square sense) for the evolution of the free energy $F$:
\begin{equation}
\mathbb{E}\bigg\{\left[\log\left(\frac{F(t_0)t_0^{\betaconst}}{F(t)t^{\betaconst}}\right)\right]^2\bigg\} \leq k(t-t_0)\ln(t/t_0).
\label{eq:beta_model3}
\end{equation}
Therefore, in this article, the values of $\betaconst$ and $k$ (as determined by numerical simulation) are key.  Specifically,
if it can be shown that $\beta_0\leq 1/3$, and that $k$  tends to zero (in some appropriate sense), then the pointwise bound~\eqref{eq:kohn_maybe} will hold, almost surely.

These questions provide the layout for the paper.  However, we also address other related issues: in Section~\ref{sec:CH} we look not only at asymmetric mixtures and the Cahn--Hilliard equation (which nicely maps on to LSW theory), but also at symmetric mixtures.  Finally, concluding remarks are  presented in Section~\ref{sec:conc}.

\section{LSW Theory Revisited}
\label{sec:LSW}

In this section we revisit LSW theory.  We first of all re-derive the standard results for completeness.  These are: the evolution equation for the radius of an individual droplet, and the self-similar dropsize distribution function, valid as $N\rightarrow \infty$; here $N$ denotes the number of droplets initially present in the system.  Re-deriving these standard results enables us to derive expressions for:
\begin{itemize}[noitemsep]

\item The probability distribution function $p_\alpha(a,t)$ for the growth rate $\alpha_i=(t/R_i)(\mathd R_i/\mathd t)$ of an individual droplet; the probability distribution function is valid in the limit as $N\rightarrow\infty$.  The existence of an analytical formula for $p_\alpha(a,t)$ establishes that it is legitimate to consider not only droplet radii as a random variable, but also, the corresponding droplet growth rates.

\item An expression for the free energy $F$ of the droplet population, and hence, a thorough understanding (for Model 1) of the quantity $\beta=-(t/F)(\mathd F/\mathd t)$.
\item An expression for the probability distribution function of $p_\beta^N(b,t)$, valid in the limit as $N\rightarrow \infty$.
\end{itemize}

\subsection{Review of standard LSW Theory}

The starting-point for LSW Theory is the mean-field solution of the Mullins--Sekerka dynamics~\eqref{eq:ms}, valid for the case of a very dilute droplet population: here, the aim is to find a highly simplified expression for the chemical potential, which will be constant (in space) in the far field, and encode the effect of all other droplets on a particular droplet $B_i$ (with $i\in\{1,\cdots,N\}$).  
As such, we solve for $\mu$ with the following constraints: 
\begin{equation}
\mu\,\begin{cases} \text{is harmonic}, & \text{for } R_i\neq |\vecx-\vecx_i|\ll d,\\
                   =1/R_i,             & |\vecx-\vecx_i|=R_i,\\
                                     \approx  \overline{u}, & \text{for}R_i\ll |\vecx-\vecx_i|\ll d.
                                    \end{cases}
\label{eq:bc}
\end{equation}
Here, $d$ is the typical distance between droplets; this is assumed to be large in comparison with $R_i$ -- this assumption is valid in the limit of very dilute systems.  
The fundamental solution is therefore given by
\[
\mu=\frac{a}{|\vecx-\vecx_i|}+b,
\]
where $a$ and $b$ are constants of integration.  
These are chosen so as to satisfy the boundary conditions~\eqref{eq:bc}, hence
\[
\mu=\frac{1}{|\vecx-\vecx_i|}R_i\left(\frac{1}{R_i}-\overline{u}\right)+\overline{u}.
\]
We also compute 
\[
\left(\frac{\partial \mu}{\partial r}\right)_{|\vecx-\vecx_i|=R_i}=-\frac{1}{R_i}\left(\frac{1}{R_i}-\overline{u}\right).
\]
Using~\eqref{eq:ms3} for spheres, this becomes
\begin{equation}
\frac{\mathd R_i}{\mathd t}=H(R_i)\left(-\frac{1}{R_i^2}+\frac{\overline{u}}{R_i}\right).
\label{eq:dRi}
\end{equation}
The Heaviside step function $H(R_i)$ is added as a pre-factor in Equation~\eqref{eq:dRi} -- this is both a regularization and a book-keeping procedure to take account of droplets whose radius shrinks to zero.
For the avoidance of doubt, we recall that $H(R_i)=0$ if $R_i\leq 0$ and $H(R_i)=1$ otherwise.

The value of the mean field $\overline{u}$ can now be obtained by imposing the constancy of the mass fraction, hence, the constancy of the total volume $\sum_{i=1}^N (4/3)\pi R_i^3$, hence
\begin{equation}
\overline{u}=\frac{\sum_{i=1}^N H(R_i)}{\sum_{i=1}^N H(R_i)R_i}.
\label{eq:uval}
\end{equation}

Next, the LSW theory is introduced in concrete terms as the limiting case of the mean-field MS theory where the number of droplets $N$ goes to infinity, while at the same time, the volume fraction
\begin{equation}
\epsilon=\lim_{\substack{ N\rightarrow\infty\\|\Omega|\rightarrow\infty}}\left(\frac{\sum_{i=1}^N (4/3)\pi R_i^3}{|\Omega|}\right)
\label{eq:Ninf}
\end{equation}
remains finite.  
The number density of droplets $P(r,t)$ is introduced, such that $P(r,t)\mathd V$ is the number of droplets in a small region of space $\mathd V=r^2\mathd r\,\mathd \omega_D$  (here, $\mathd\omega_D$ is the differential solid-angle element in $D$ dimensions).
Using standard conservation-type arguments, the evolution equation for $P(r,t)$ is just
\[
\frac{\partial P}{\partial t}+\nabla\cdot\left(\bm{v} P\right)=0,
\]
where $\bm{v}$ is the velocity of one of the droplets.  
The problem is radially symmetric, hence only the radial velocity is required.  
This is known from Equation~\eqref{eq:dRi}, hence
\[
v_r(r,t)=-\frac{1}{r^2}+\frac{\overline{u}(t)}{r}.
\]
Using the expression for divergence in spherical polar coordinates, for a radially-symmetric configuration, the evolution equation for $P$ becomes:
\[
\frac{\partial P}{\partial t}+\frac{1}{r^2}\frac{\partial }{\partial r}\left(r^2 v_r P\right)=0.
\]
If we define
\[
P(r,t)r^2=p(r,t), \text{ such that }\int_0^\infty P(r,t)r^2\mathd r=\int_0^\infty p(r)\mathd r,
\]
then the required evolution equation is
\begin{equation}
\frac{\partial p}{\partial t}+\frac{\partial }{\partial r}\left( v_r p\right),\qquad v_r(r,t)=-\frac{1}{r^2}+\frac{\overline{u}}{r}.
\label{eq:fp}
\end{equation}
Hence, 
\begin{equation}
p(r,t)=\text{Numer of droplets with radius between }r\text{ and }r+\mathd r, \text{ at time }t.
\label{eq:p_def}
\end{equation}
%
In analogy with Equation~\eqref{eq:uval}, Equation~\eqref{eq:fp} is closed by requiring:
\[
\overline{u}=\frac{\int_0^\infty p(r,t)\mathd r}{\int_0^\infty rp(r,t)\mathd r}.
\]
We now seek a similarity solution of Equation~\eqref{eq:fp}.  We write
\begin{equation}
p=t^a f(x),\qquad x=\frac{r}{ct^b}.
\label{eq:sim}
\end{equation}
We fix $a$ in the first instance.  We use the fact that the volume fraction $\epsilon$ is constant, to compute 
\begin{eqnarray*}
\epsilon&=&\frac{1}{|\Omega|}\iiint_{\Omega}(4\pi/3)\pi r^3 P(r)\mathd V,\\
        &=&\frac{1}{|\Omega|}\iiint_{\Omega}(4\pi/3)\pi r^3 p(r)\,\mathd r\,\mathd\omega_D,\\
                &=&\tfrac{1}{3}|\Omega|^{-1}(4\pi)^2\int_0^{R_{\mathrm{max}}} r^3 p(r)\mathd r.
\end{eqnarray*}
Here, $R_\mathrm{max}$ is a notional cutoff, with $R_{\mathrm{max}}\rightarrow\infty$ along with $|\Omega|\rightarrow\infty$, in such a way that $\epsilon$ remains finite.
Hence,
\[
\epsilon=\tfrac{1}{3}|\Omega|^{-1}(4\pi)^2 c^4 t^{a+4b}\int_0^{x_{\mathrm{max}}} x^3 f(x)\mathd x,
\]
where again, $x_\mathrm{max}\rightarrow\infty$ is a notional cut--off, chosen such that $\epsilon$ remains finite as $|\Omega|\rightarrow\infty$.
Thus, in order for $\epsilon$ to remain constant, it is required that $a=-4b$.
We now substitute the similarity solution~\eqref{eq:sim} into Equation~\eqref{eq:fp}.  After manipulations, we obtain:
\begin{equation}
-\tfrac{1}{3}c^3\left[3f(x)+ f'(x)x\right]+\frac{\partial}{\partial x}\left[\left(-\frac{1}{x^2}+\frac{\hat{u}}{x}\right)f\right]=0,\qquad
\hat{u}=\frac{\int_0^\infty f(x)\mathd x}{\int_0^\infty xf(x)\mathd x}.
\label{eq:ls_sim}
\end{equation}
Following convention, we write $(1/3)c^3=\scalegamma$, to give:
\begin{equation}
-\gamma\left[3f(x)+ f'(x)x\right]+\frac{\partial}{\partial x}\left[\left(-\frac{1}{x^2}+\frac{\hat{u}}{x}\right)f\right]=0,\qquad
\hat{u}=\frac{\int_0^\infty f(x)\mathd x}{\int_0^\infty xf(x)\mathd x}.
\label{eq:ls_sim1}
\end{equation}
This can then be integrated to give~\cite{bray2002theory}:
\begin{equation}
\ln [f(x)]=\int^x \frac{\mathd y}{y}\frac{2-y-3\scalegamma y^3}{\scalegamma y^3-y+1}.
\label{eq:ls_sim2}
\end{equation}
Equation~\eqref{eq:ls_sim1} gives a family of potential solutions, all parametrized by $\scalegamma$.  
The equation is also potentially without a normalizable solution with $f(x)\rightarrow 0$ as $x\rightarrow \infty$.  
These problems are solved by imposing two conditions on Equation~\eqref{eq:ls_sim2}:
\begin{itemize}[noitemsep,topsep=0.5pt]
\item To ensure normalizability, the solution $f(x)$ should have compact support;
\item The value $\scalegamma=4/27$ must be selected.
\end{itemize}
The rationale for the second condition is related to the fixed points of Equation~\eqref{eq:ls_sim} and was determined by Lifshitz and Slyzov~\cite{lifshitz1961kinetics} and summarized by Bray~\cite{bray2002theory}.
As such, the following solution for $f(x)$ is found in three dimensions ($D=3$), by integration of Equation~\eqref{eq:ls_sim2}:
\begin{equation}
f(x)=\begin{cases} \text{Const.}\times x^2(3+x)^{-1-4D/9}\left(\tfrac{3}{2}-x\right)^{-2-5D/9}\exp\left(-\frac{D}{3-2x}\right),& 0\leq x<(3/2),\\
0,&\text{otherwise}.
\end{cases}
\label{eq:fx_final}
\end{equation}
In LSW theory, the expected mean radius is computed as follows:
%
%
\begin{equation}
\langle R\rangle 
= \frac{\int_0^\infty r p(r)\,\mathd r}{\int_0^\infty p(r)\,\mathd r}
= c t^{1/3}\frac{\int_0^\infty x f(x)\,\mathd x}{\int_0^\infty f(x)\,\mathd x}
=(3\scalegamma t)^{1/3},
\label{eq:R_av}
\end{equation}
where the last equation follows since $c=(3\scalegamma)^{1/3}$ and since the distribution $f(x)$ in Equation~\eqref{eq:fx_final} has the property
$\left[\int x f(x)\mathd x\right]/\left[\int f(x)\mathd x\right]=1$.  We note also that
\begin{equation}
\overline{u}=1/\langle R\rangle, \text{hence }\overline{u}=(3\scalegamma t)^{-1/3}.
\label{eq:u_av}
\end{equation}


\subsection{Convergence to the self-similar distribution function}

It is not guaranteed that an arbitrary initial configuration of droplets will enable the dropsize distribution function $p(r,t)$ to the self-similar form in Equation~\eqref{eq:fx_final}.
It is intuitively obvious that if such convergence is to be achieved, the initial condition for $p(r,t)$ should have compact support.
Further criteria on the initial non-self-similar distribution $p(x,t=0)$ are required.  These are the so-called weak selection rules~\cite{giron1998weak}.
 If the weak selection rules for the initial condition are not satisfied, the late-time behaviour may become non-self-similar (e.g. Reference~\cite{niethammer1999non}).  An overview of this problem is also provided in Reference~\cite{mielke2006analysis}.  

\subsection{LSW theory -- Droplet growth rates in the self-similar regime}

We now develop an analytical formula for the probability distribution function of the growth rates of individual droplets, valid in the limit as $N\rightarrow \infty$, and under assumption that the dropsize distribution function converges to the self-similar form~\eqref{eq:fx_final}. The existence of this formula is a partial answer to Question 1 in the introduction, in the sense that this formula establishes that the probability distribution function of droplet growth rates is a legitimate object of study.

For these  purposes, we introduce the droplet growth rate
\begin{equation}
\alpha_i= \frac{t}{R_i}\frac{\mathd R_i}{\mathd t}= \frac{t}{R_i}\left(-\frac{1}{R_i^2}+\frac{\overline{u}}{R_i}\right).
\label{eq:alphai_def}
\end{equation}
We now want to characterize the distribution of the $\alpha_i$'s, over all droplets, which we denote by $p_\alpha(a,t)$:
\[
p_\alpha(a,t)=\text{Number of droplets with growth rate }a\text{ in the range }a\leq \alpha \leq a+\mathd a,\text{ at time }t.
\]  
Formally, we can calculate $p_\alpha(a,t)$ using the standard change-of-variable formula for probability theory:
\[
p_\alpha(a,t)=p(r(a),t)\left|\frac{\partial r}{\partial a}\right|,
\]
where $a$ and $r$ are connected by
\begin{equation}
a=\alpha,\qquad \alpha=\frac{t}{r}\left(-\frac{1}{r^2}+\frac{\overline{u}}{r}\right).
\label{eq:ar}
\end{equation}
As such, we have the following formal identity (the $t$-dependence is  suppressed for now):
%
%
\begin{equation}
p_\alpha(a)=p(r(a))\left|\frac{\partial r}{\partial a}\right|
         =p(r(a))\left|\frac{\partial a}{\partial r}\right|^{-1}.
\label{eq:change}
\end{equation}
It is not straightforward to implement the substitutions in Equation~\eqref{eq:change}, because $a$ is not a monotone function of $r$ (e.g. Figure~\ref{fig:r1} on the entire range of the function $a(r)$).  
Indeed, the derivative of $a(r)$ changes sign at $r=(3/2)\overline{u}^{-1}$, i.e. $\partial a/\partial r=0$ at $r=(3/2)\langle R\rangle$.  
Before solving this problem, we remark that the existence of the local maximum $\partial a/\partial r=0$ gives rise to the following useful results:
\begin{proposition}
There is a maximum droplet growth rate
\begin{equation}
\alpha_{max}=\tfrac{1}{3}\left(\frac{4}{9}\frac{t}{\langle R\rangle^3}\right),
\label{eq:alpha_max1}
\end{equation}
where $\langle R\rangle$ is computed via Equation~\eqref{eq:R_av}.
\end{proposition}
\begin{proof}
We start with Equation~\eqref{eq:ar}.  We compute $\partial\alpha/\partial r$ and set the result to zero.  This gives 
$r=(3/2)\langle R\rangle$ for the maximum growth rate.  This particular value of $r$ is then substituted back into 
Equation~\eqref{eq:ar} to produce $\alpha_{max}=(1/3)[(4/9)(t/\langle R\rangle^3)]$, as required.
\end{proof}
%
Furthermore,
\begin{proposition}
If the dropsize distribution function has the self-similar form~\eqref{eq:fx_final}, the maximum droplet growth rate simplifies:
\[
\alpha_{max}=1/3.
\]
\end{proposition}
The proof of this statement is by direct computation, specifically by substituting $\langle R\rangle=(4t/9)^{1/3}$ into Equation~\eqref{eq:alpha_max1}.

Having now established the existence of the maximum droplet growth rate, it follows that that $a(r)$ is non-monotonic, and hence, the formal change-of-variables law~\eqref{eq:change} needs clarification.  
Therefore, to calculate $p_\alpha(a)$ properly, we refer to Figure~\ref{fig:r1}. 
\begin{figure}
    \centering
        \includegraphics[width=0.6\textwidth]{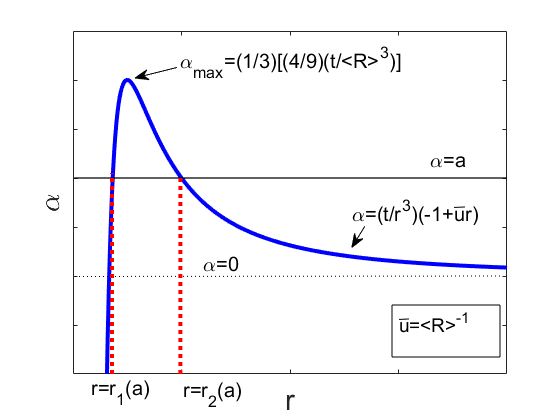}
        \caption{Definition sketch for the change of variable $a=(t/r^3)(-1+\overline{u}r)$.}
    \label{fig:r1}
\end{figure}
We look at a definite fixed value of $a$, denoted by $a$.  
For $a>0$, we read off the definitions of $r_1(a)$ and $r_2(a)$ from the figure.  
We look at the cumulative probability function for $a$,
\begin{eqnarray*}
F_\alpha(a)&=&\mathbb{P}(\alpha \leq a),\\
           &=&\mathbb{P}(r\leq r_1(a))+\mathbb{P}(r\geq r_2(a)),\\
                     &=&F_r(r_1(a))+\left[1-F_r(r_2(a))\right].
\end{eqnarray*}
We differentiate to compute the probability distribution function:
\begin{eqnarray*}
p_\alpha(a)&=&\frac{\mathd F_\alpha(a)}{\mathd a},\\
           &=&\frac{\partial F_r}{\partial r}\bigg|_{r_1(a)}\frac{\mathd r_1}{\mathd a}
                    -\frac{\partial F_r}{\partial r}\bigg|_{r_2(a)}\frac{\mathd r_2}{\mathd a},\\
                     &=&p(r_1(a))\frac{\mathd r_1}{\mathd a}-
                    p(r_2(a))\frac{\mathd r_2}{\mathd a}.
\end{eqnarray*}
Thus, the distribution of exponents $a$ is established for $a>0$:
\begin{equation}
p_\alpha(a)=p(r_1(a))\left|\frac{\mathd r_1}{\mathd a}\right|+
                    p(r_2(a))\left|\frac{\mathd r_2}{\mathd a}\right|,\qquad a>0,
\label{eq:change_of_var_final}
\end{equation}
where the first instance of $\left|\cdot\right|$ is added just to make the formula appear more symmetric.  
Referring back to Figure~\ref{fig:r1}, at $a=0$, the two roots $r_1(a)$ and $r_2(a)$ coincide, and for $a<0$ only one root (denoted by $r_1(a)$ survives).  
As such, the following final form of $p_\alpha$ applies,
\begin{equation}
p_\alpha(a,t)=\begin{cases}p(r_1(a),t)\left|\frac{\mathd r_1}{\mathd a}\right|+
                    p(r_2(a),t)\left|\frac{\mathd r_2}{\mathd a}\right|,& a>0,\\
                     p(r_1(a),t)\left|\frac{\mathd r_1}{\mathd a}\right|,&a\leq 0.
                    \end{cases}
\label{eq:change_of_var_final_final}
\end{equation}
where we have restored the time-dependence of the distributions.

In the case where $p(r,t)$ satisfies the LSW distribution~\eqref{eq:fx_final}, it is possible to compute the corresponding growth-rate distribution $p_\alpha(a,t)$.  
This  is shown in Figure~\ref{fig:p_alpha}.  
It is verified that in this instance, the distribution of growth rates is time-dependent.
\begin{figure}
    \centering
        \includegraphics[width=0.6\textwidth]{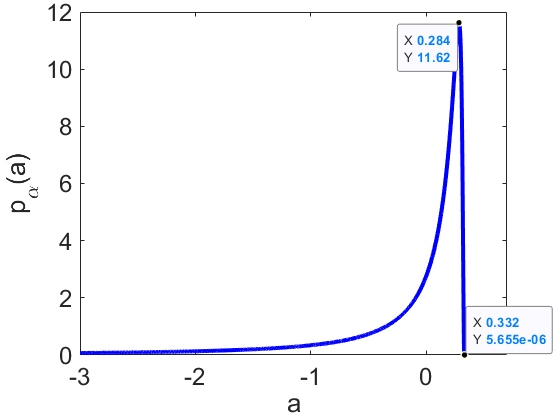}
        \caption{The probability distribution function $p_\alpha(a)$.  The distribution is time-independent so the notation $p_\alpha(a,t)$ can be replaced with $p_\alpha(a)$.}
    \label{fig:p_alpha}
\end{figure}
The distribution $p_\alpha(a)$ goes to zero at $a=1/3$, corresponding to the fact that $a=1/3$ is the maximum droplet growth rate.  
Otherwise, the distribution is sharply peaked at $a\approx 0.284$, corresponding to the mode or the most probable droplet growth rate in the problem.  
Notably, the distribution is strongly skewed to the left, with a long tail of negative growth rates extending to $a=-\infty$.  
The negative growth rates correspond to the evaporating droplets.  
The maximum growth rate corresponds to a winner-takes-all scenario where a single droplet is growing to the maximal extent possible, at the expense of all other droplets in the system.  
Hence, the distribution in Figure~\ref{fig:p_alpha} makes sense physically.

\subsection{System energy}
\label{sec:systemEnergy}

We now develop an energy function $F$ to characterize the dynamics of Equation~\eqref{eq:dRi}.    This then enables us to compute $\beta=-(t/F)(\mathd F/\mathd t)$.   The formula is completely general, however, in the case where $N\rightarrow \infty$ and where $p(r,t)$ assumes the self-similar form, we show that the probability distribution function of $\beta$ is well-defined and equaly (trivially) to a delta function.  
 This then provides the necessary insights to argue for the existence of a probability distribution function for $\beta$ for finite-size (and non-self-similar) systems.  In this way, {\textbf{Question 1}} in the introduction is answered.

The starting-point for the development of the energy function is the identification $F\sim \sigma \sum_{i=1}^N 4\pi R_i^2$, where $\sigma$ is a surface tension, and the summation gives the total interfacial area of the system.  Implicit in the adimensional Equation~\eqref{eq:dRi} is the value $\sigma=1/2$.  
Implicit also in that equation is the constraint that the total droplet volume is constant, $\sum_{i=1}^N (4\pi /3)R_i^2=V_0$, where $V_0$ is constant.  
Thus, properly constituted, the surface energy contains a constraint term:
\begin{equation}
F= \sum_{i=1}^N \tfrac{1}{2} R_i^2-\lambda\left(\sum_{i=1}^3 \tfrac{1}{3} R_i^3-\frac{V_0}{4\pi}\right),
\label{eq:Fdef}
\end{equation}
where $\lambda$ is the possibly time-dependent Lagrange multiplier which enforces the constancy of $\sum_{i=1}^3 (4\pi/3)R_i^3$, and where we have omitted an overall factor of $4\pi$ in the definition of $F$ -- this is done for convenience. 

By differentiating Equation~\eqref{eq:Fdef}, we obtain:
\begin{equation}
\frac{\mathd F}{\mathd t}=\sum_{i=1}^N \left(R_i-\lambda R_i^2\right)\dot R_i+(\mathd \lambda/\mathd t)\left(\sum_{i=1}^3 \tfrac{1}{3} R_i^3-\frac{V_0}{4\pi}\right)
\end{equation}
The last term proportional to $\mathd \lambda/\mathd t$ vanishes on enforcing the constraint on $\sum_{i=1}^3 (4\pi/3)R_i^3$.  
Thus,
\begin{equation}
\frac{\mathd F}{\mathd t}
\stackrel{\text{Eq.~\eqref{eq:dRi}}}{=}\sum_{i=1}^N H(R_i)\left(R_i-\lambda R_i^2\right)\left(-\frac{1}{R_i^2}+\frac{\overline{u}}{R_i}\right)
\stackrel{\lambda=\overline{u}}{=}-\sum_{i=1}^N \frac{H(R_i)}{R_i}\left(1-\overline{u}R_i\right)^2,
\label{eq:dFdt}
\end{equation}
hence $\mathd F/\mathd t\leq 0$.
Here, the equation $\lambda=\overline{u}$ can be made, since $\lambda$ and $\overline{u}$ are associated with the same constraint.  
%
Using the identity $\lambda=\overline{u}$, we can write   $\partial F/\partial R_i=R_i-\overline{u} R_i^2$, and hence, from Equation~\eqref{eq:dRi},
\begin{equation}
\frac{\mathd R_i}{\mathd t}=-m(R_i)\frac{\partial F}{\partial R_i},\qquad m(R_i)=H(R_i)R_i^3.
\label{eq:dRi_m}
\end{equation}
Thus, the dynamics of the droplets take the form of a gradient flow, with mobility $m(R_i)=H(R_i)R_i^3$.  Furthermore, we can therefore write
\begin{equation}
\frac{\mathd F}{\mathd t}=-\sum_{i=1}^N m(R_i)\left(\frac{\partial F}{\partial R_i}\right)^2,
\label{eq:dFdt1}
\end{equation}
which makes the relation $\mathd F/\mathd t\leq 0$ more manifest.
In analogy to the growth rate $\alpha_i$ (Equation~\eqref{eq:alphai_def}) for individual droplets, we introduce an energy decay rate, applicable to the entire system of $N$ droplets:
\begin{equation}
\beta=-\frac{t}{F}\frac{\mathd F}{\mathd t}.
\label{eq:betaN}
\end{equation}

In LSW theory, $\beta$ necessarily takes on only one value; in other words, the distribution of $\beta$-values in that limit is the delta function.  This is noted in the following proposition:
\begin{proposition}
\label{prop:awesome}
At late times, $p_\beta(b,t)\rightarrow \delta(b-(1/3))$, for the LSW limit, i.e. for the large-domain limit~\eqref{eq:Ninf} and the self-similar dropsize distribution~\eqref{eq:fx_final}.
\end{proposition}

\begin{proof}
Once the volume-constraint $\sum_{i=1}^N (4\pi/3)R_i^3=V_0$ has been implemented, the energy is just $F=(1/2)\sum_{i=1}^N R_i^2$. 
In the LSW limit, this can be computed explicitly, via Equation~\eqref{eq:sim}
\begin{equation}
F=
\tfrac{1}{2}\int_0^\infty r^2 p(r,t)\mathd r.
\label{eq:F_expected}
\end{equation}
%
We therefore have:
\[
F=\tfrac{1}{2}c^3 t^{3b+a}\int_0^\infty x^2 f(x)\mathd x,\qquad \text{as }t\rightarrow\infty.
\]
The late-time limit is required here as the LSW theory is valid only asymptotically, as $t\rightarrow\infty$.  
We also use $a=-4b$, hence
\[
F=\tfrac{1}{2}c^3 t^{-b}\int_0^\infty x^2 f(x)\mathd x.
\]
We now use $b=1/3$ to conclude that $F\propto t^{-1/3}$, and hence, $\beta=-(t/F)(\mathd F/\mathd t)=1/3$.  
Thus, $\beta$ takes only a single value in the LSW theory, hence $p_\beta(b,t)\rightarrow\delta(b-(1/3))$ as $t\rightarrow\infty$.
\end{proof}

In contrast to the LSW theory considered in this section, in Section~\ref{sec:drop_pop}, equations for $N<\infty$ droplets (i.e. Equations~\eqref{eq:dRi}--\eqref{eq:uval}) are  solved via numerical simulation.  An ensemble (with $M$ members) of such simulations is constructed, and a probability distribution function $p_\beta^N(b,t)$ is thereby constructed.
Proposition~\ref{prop:awesome} shows that $p_\beta^N(b,t)\rightarrow \delta(b-(1/3))$ as $N\rightarrow \infty$, provided the distribution of initial droplet radii satisfies the weak selection rules.  However, for finite $N$ and / or for initial conditions not satisfying the weak selection rules, the possibility is open that the distribution of $p_\beta^N(b,t)$ may be broad.   This is examined in depth in the next section.

\section{Model 1-- Numerical Simulations}
\label{sec:drop_pop}


\subsection{Methodology}

In this section, we solve Equation~\eqref{eq:dRi} numerically.  The initial radius of any particular droplet is given by
\begin{equation}
R_i(t=0)=r_i,
\label{eq:dRi_init}%
\end{equation}%
where $r_i$ is a random variable drawn from a uniform distribution between $0$ and $1$.  The system of equations~\eqref{eq:dRi} is solved numerically using ODE45 in Matlab.  
Equation~\eqref{eq:dRi} has a coordinate singularity at $R_i=0$; this is regularized in the numerical method by solving 
\begin{equation}
\frac{\mathd R_i}{\mathd t}=H(R_i+\epsilon)\left(-\frac{1}{R_i^2}+\frac{\overline{u}}{R_i}\right)
\end{equation}
instead; here $H(s)$ is the Heaviside step function.  In the simulations, we have taken $\epsilon=10^{-3}$; however, we have also verified that reducing $\epsilon$ to $10^{-4}$ makes no change to the results.
This regularization correctly reduces the radius of a small droplet to $0$ and allows us to treat the coordinate singularity numerically.

\subsection{Results -- Single Simulation}

We first of all show a space-time plot of the instantaneous histogram of $x$-values for a large droplet 
population $(N=100,000$), generated from a single simulation.  
The variable $x$ is recalled here as $x=R/c t^{1/3}$, where $R$ is a droplet radius.  The purpose of this calculation is to establish the extent to which the simulation results agree with LSW theory.
%
%
%
%
The results are presented in Figure~\ref{fig:histoLong}.
\begin{figure}[htb]
	\centering
		\includegraphics[width=0.7\textwidth]{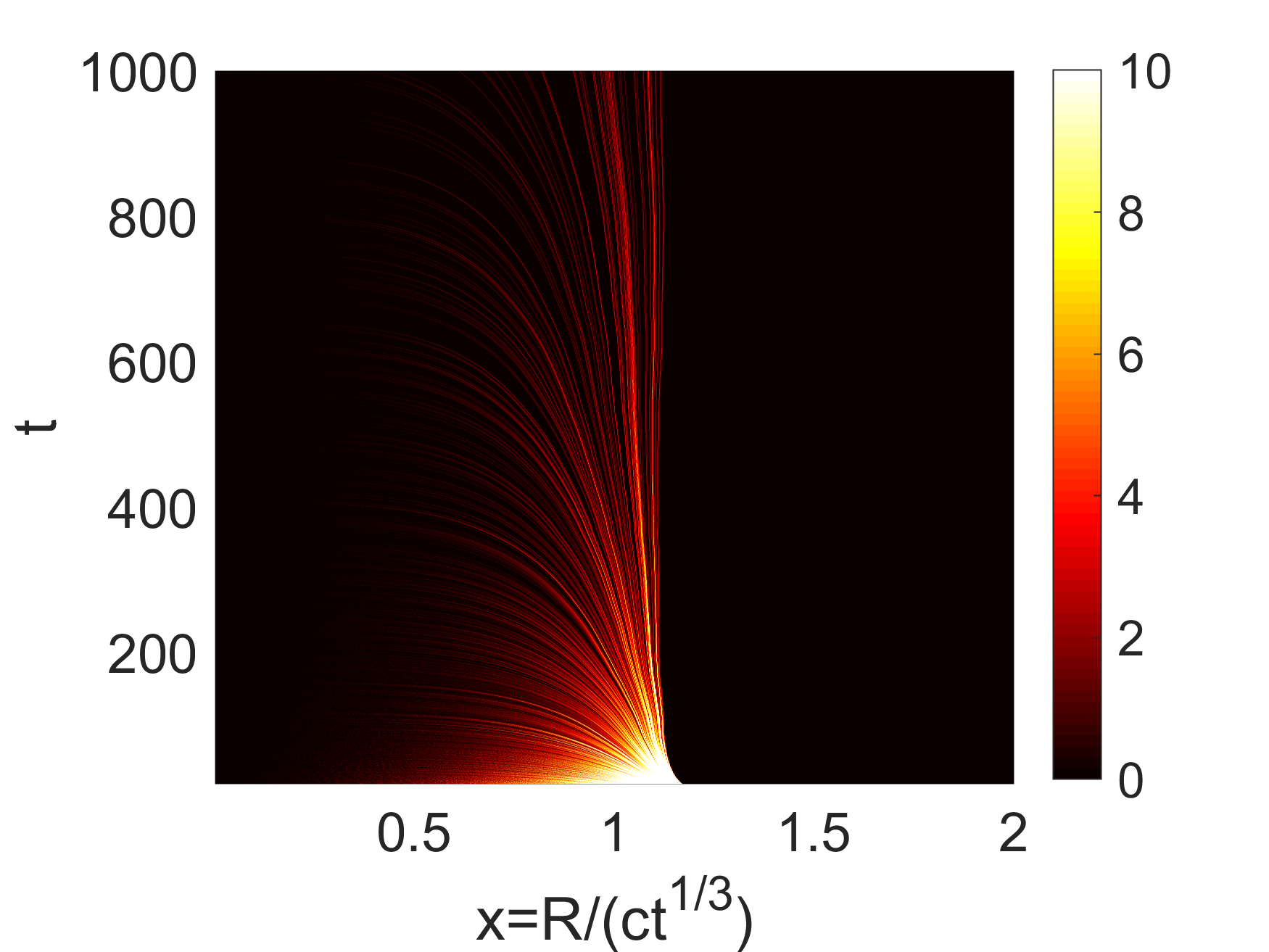}
		\caption{Plot of the instantaneous histogram of the dropsize distribution, with $t\geq 50$,  for a single simulations with $N=100,000$ droplets present initially.  The plot uses the self-similar coordinate $x=R/ct^{1/3}$.}
	\label{fig:histoLong}
\end{figure}
The first observation is that the histogram does not reach a statistically steady state.  Two reasons for this non-convergence are identified.

\paragraph{Extreme Finite-Size Effects:} Since the system is finite, the total initial droplet volume is finite.  Therefore, large droplets can draw on only a finite number of smaller droplets as fuel for the coarsening.  The coarsening phenomenon will therefore eventually cease when there is only a single large droplet present.  This is an extreme finite-size effect which prevents the system from attaining the LSW statistically steady state at extremely late times.

\paragraph{Dependence on initial drop-size distribution:} At late times (but before the onset of the extreme finite-size effects), the system is still far from the steady state, as can be seen by inspection of Figure~\ref{fig:histoLong}.  Therefore, it can be concluded that the statistics of the system do not attain the LSW form even at such late times.  
To emphasize this point,  in Figure~\ref{fig:n100000}(a) we make a comparison between LSW theory and the numerical results  by plotting cumulative histograms of droplet radii as a function of time, and comparing with LSW theory -- the two sets of curves visibly disagree. 

For clarity's sake, we summarize here the method of generating the cumulative histograms in Figure~\ref{fig:n100000}.  
These are extracted from a single simulation, with data extracted in a time range before the onset of extreme finite-size effects. 
As such,  the cumulative histogram  at time $t$ means that we bin all values of $x$ recorded across all droplets in the simulation, starting at $t=20$, up to and including and including the final time $t$, where $t\leq 40$.  
This can be contrasted with the instantaneous histogram, which would be obtained by binning all values of $x$ recorded across the simulations at precisely the time $t$.  
If the system reaches a statistically steady state, the two histograms should agree at late times.   
\begin{figure}[htb]
	\centering
	\subfigure[]{\includegraphics[width=0.45\textwidth]{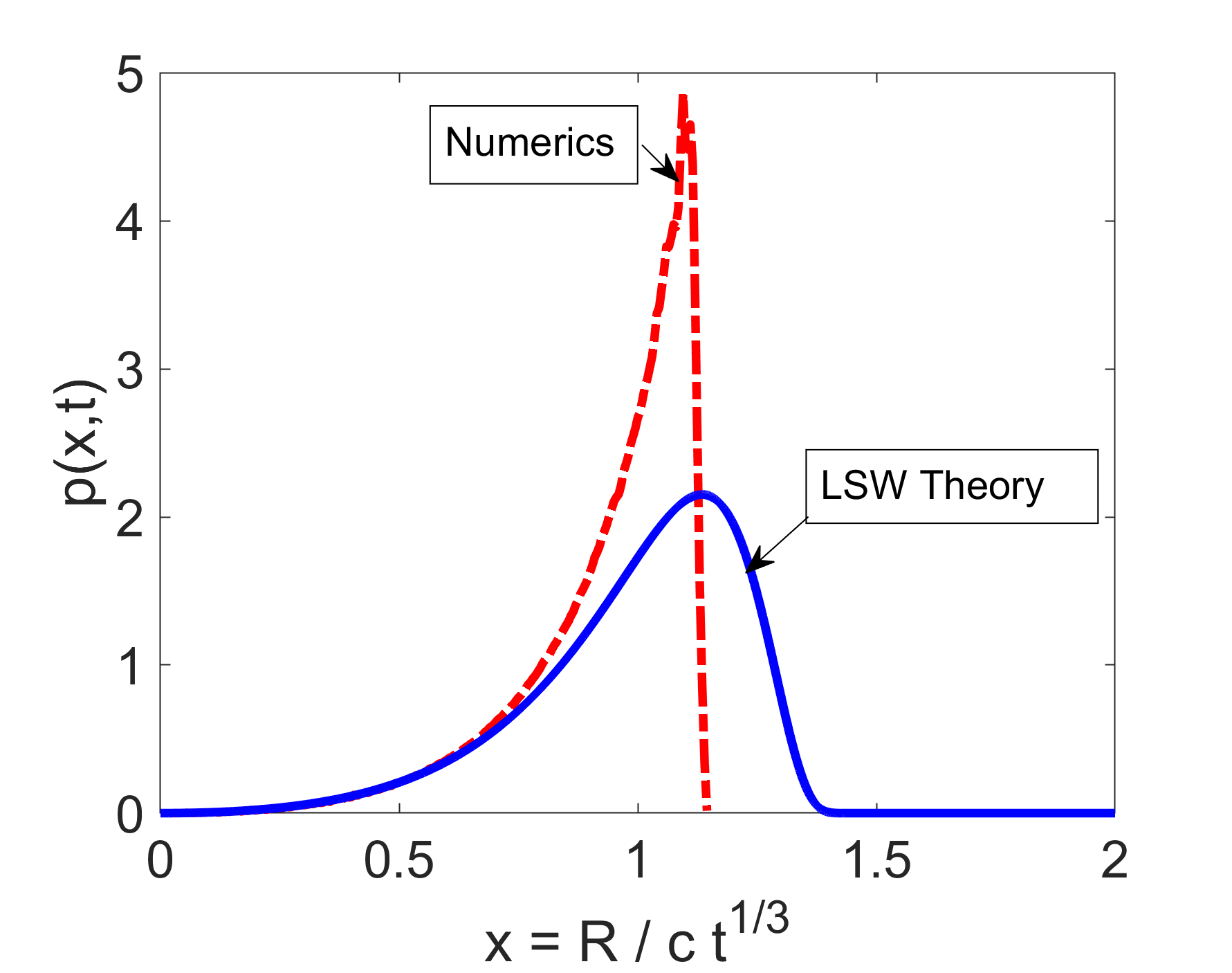}}
	\subfigure[]{\includegraphics[width=0.45\textwidth]{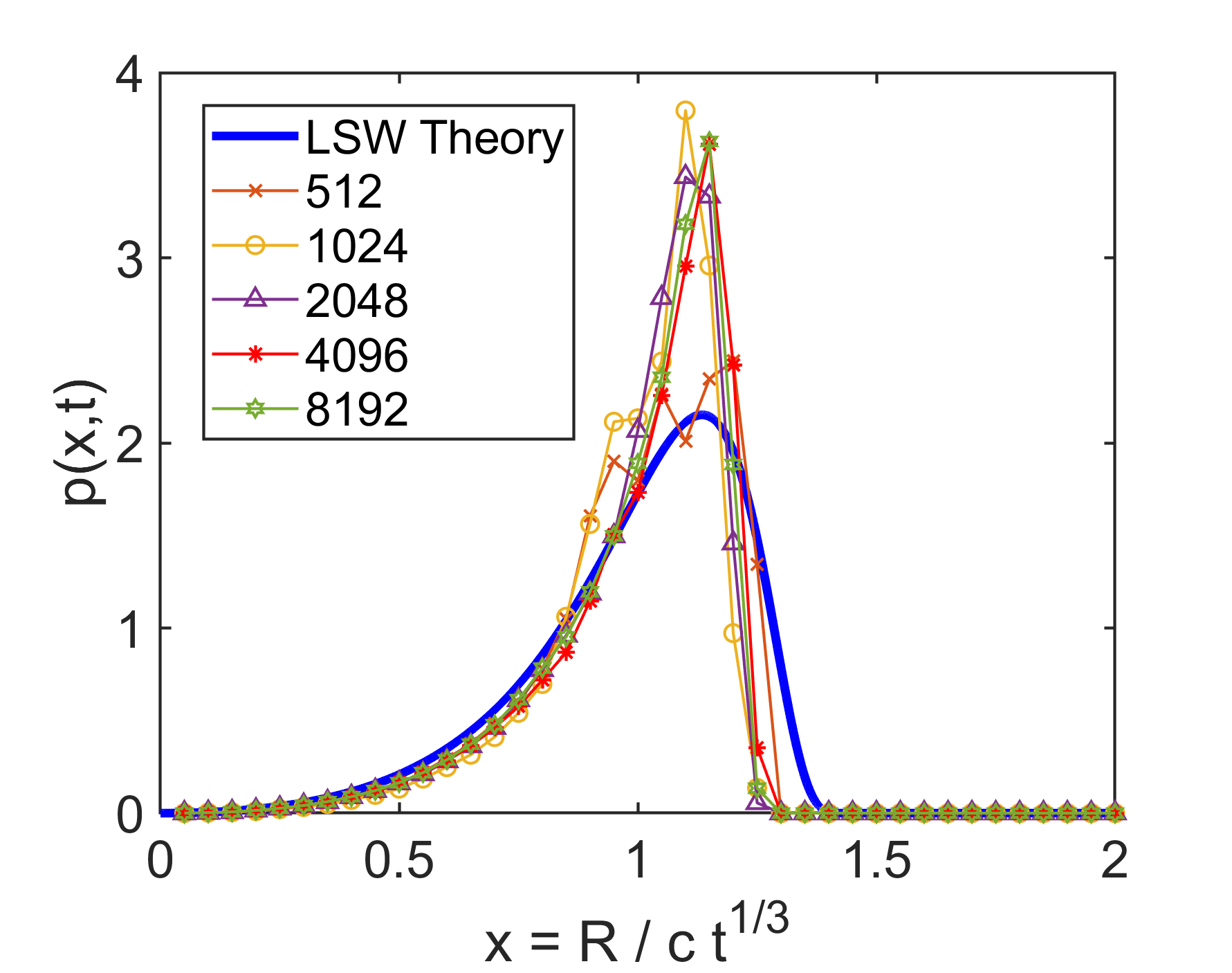}}
	\caption{Comparison between numerics and LSW theory.  In both panels, the numerics correspond to cumulative histograms, and the theory refers to the LSW theory with the analytic self-similar distribution.
	(a)  Numerics, with $N=100,000$ and $40\leq t\leq 200$.   (b)  Numerics, various values of $N$ and $10\leq t\leq 40$. }
	\label{fig:n100000}
\end{figure}

It is of interest to look into the lack of agreement between the numerics and the LSW theory in Figure~\ref{fig:n100000}.  
The number of droplets $N$ present initially in the simulation can be ruled out as the cause of the disagreement: the dependence of the cumulative histogram on $N$ is shown in Figure~\ref{fig:n100000}(b); there is little or no difference between all of the considered $N$ values.
Therefore, the cause of the disagreement in Figure~\ref{fig:n100000} can be attributed to the shape of the initial dropsize distribution: the initial dropsize distribution is compactly supported (the uniform distribution with initial radii between $0$ and $1$); however, this distribution is not smooth at the points where it touches down to zero.  
Therefore, the initial distribution does not satisfy the weak selection rules (Section~\ref{sec:LSW}), and hence, convergence to the LSW statistics is not guaranteed: this explains the results in Figure~\ref{fig:n100000}.  
It can be emphasized that in other works on droplets (e.g. Reference~\cite{yao1993theory}), the initial dropsize distribution was carefully selected such that late-time convergence to LSW statistics was obtained.  From our results, the convergence to the LSW statistics is demonstrated not to be robust.  
For completeness, we plot the cumulative histograms of droplet growth rates in Figure~\ref{fig:alphaTheory}, where we again demonstrate non-agreement between the LSW theory and the numerical results.
\begin{figure}[htb]
\centering
\includegraphics[width=0.6\textwidth]{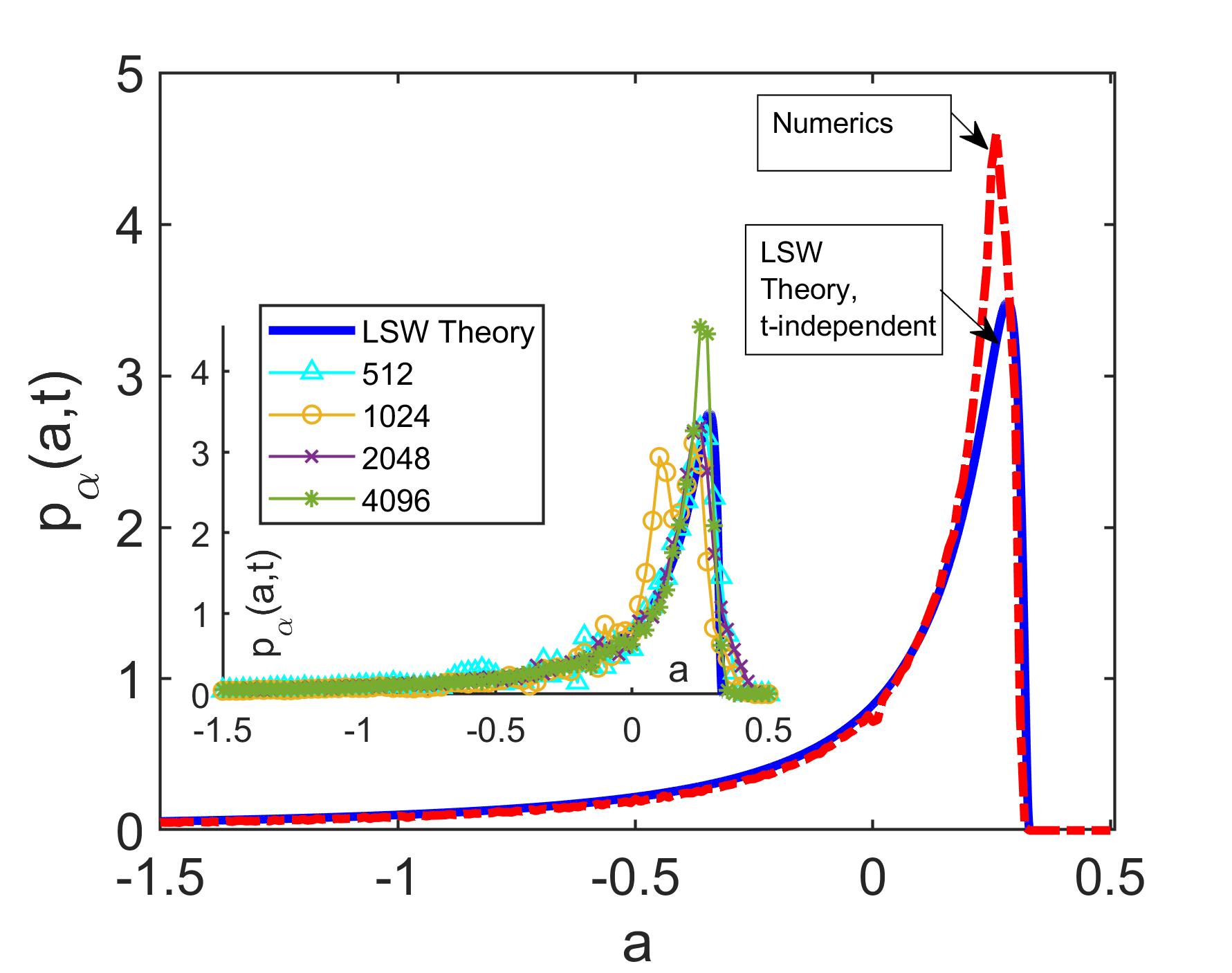}
\caption{Droplet growth rates:  Comparison between LSW theory and numerics for $N=100,000$ (the inset shows the effect of varying $N$).  For the numerics, the presented results are cumulative histograms taken over the range $10 < t < 40$.}
\label{fig:alphaTheory}
\end{figure}

For the same simulation ($N=100,000$), we examine  $\beta$, computed as in Equation~\eqref{eq:betaN}.  This is recalled to be a property of the entire droplet population, rather than a property of any one individual droplet (cf. Section~\ref{sec:LSW}).    The time series of $\beta$ is shown in Figure~\ref{fig:beta_N10000}.
From the figure, it can be seen that $\beta$ consists of a mean component and a fluctuation, which we hereafter write as $\beta=\betaconst+\delta\beta$.  The fluctuation $\delta\beta$ is a piecewise-continuous function of time, which occasionally jumps discontinuously from a positive value to either a smaller positive value or to a negative value (inset, Figure~\ref{fig:beta_N10000}).  The magnitude of the jump is seen to increase with time, along with the interval between jumps.  These observations enable us to trace the cause of the jumps: they are associated with the death of a droplet when the Heaviside step function is activated in Equation~\eqref{eq:dRi}; this in turn induces a discontinuity in $\beta$.  The waiting time between between such jumps increases as the system evolves: at late times, there are only a few droplets left, hence fewer droplet deaths and longer waiting times.  Similarly, at late times the remaining droplets are relatively large, meaning that the death of any one droplet induces a relatively large jump discontinuity in $\beta$.

\begin{figure}[htb]
    \centering
     \includegraphics[width=0.6\textwidth]{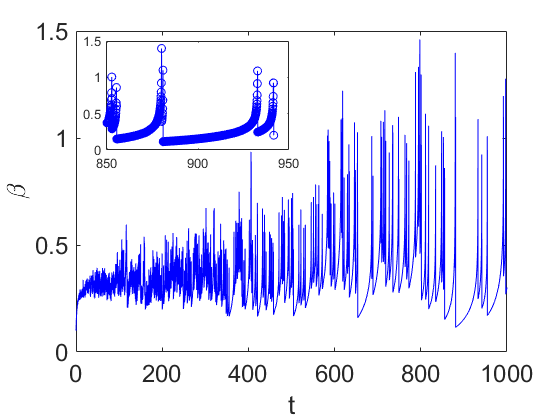}
				\caption{Time series of $\beta$ for a single simulation with $N=100,000$ droplets present initially.  The inset shows a portion of the main figure, on a larger scale.}
    \label{fig:beta_N10000}
\end{figure}

The `fitting' of a stochastic model to $\beta$ will be the subject of future work.  For the time being, it suffices to note that $\beta$ is piecewise smooth, and that the jumps are well-separated.  Therefore, the free energy $F$ may be recovered from $\beta=-(t/F)(\mathd F/\mathd t)$ via ordinary Riemann integration:
\begin{equation}
-\log \frac{F(t)}{F(t_0)}=\log(t/t_0)^{\betaconst}+\int_{t_0}^t\frac{\delta\beta}{t}\mathd t,
\label{eq:fint1}
\end{equation}
where the integral is performed piecewise, i.e. over each continuous segment of the time series of $\delta \beta$.  Equation~\eqref{eq:fint1} may be re-arranged as:
\[
\log \left(\frac{F(t_0)t_0^{\betaconst}}{F(t)t^{\betaconst}}\right)=\int_{t_0}^t\frac{\delta\beta}{t}\mathd t.
\]
Both sides can be squared to give:
\[
\left[\log \left(\frac{F(t_0)t_0^{\betaconst}}{F(t)t^{\betaconst}}\right)\right]^2=\left[\int_{t_0}^t\frac{\delta\beta}{t}\mathd t\right]^2
\leq  (t-t_0)\int_{t_0}^t\left(\frac{\delta\beta}{t}\right)^2 \mathd t,
\]
where we have used the Cauchy-Schwarz inequality with $t>t_0$.  Finally, by taking expectation values with respect to the measured induced by the random jumps in $\delta \beta$, we obtain:
\[
\mathbb{E}\bigg\{\left[\log\left(\frac{F(t_0)t_0^{\betaconst}}{F(t)t^{\betaconst}}\right)\right]^2\bigg\} \leq (t-t_0)\int_{t_0}^{t}\frac{\mathbb{E}(\delta\beta^2)}{t^2}\mathd t,\qquad t>t_0.
\]
i.e. Equation~\eqref{eq:beta_model2} in the introduction.

\subsection{Results -- Ensemble of simulations}

We now turn to simulations to estimate the expectation value of $\delta\beta^2$.  
As $\beta=\betaconst+\delta\beta$ is a characterization of an entire droplet population (i.e. an entire simulation of $N$ droplets), it is necessary to gather statistics of $\beta$ across an ensemble of many such simulations -- see Table~\ref{tab:ensembles}.
For these purposes, parallel computing in Matlab is implemented such that there is one realisation of the ensemble per CPU core.  
\begin{table}[htb]
\centering
\begin{tabular}{|c|c|c|} 
\hline
    Ensemble  & $N$ & Number of simulations  \\
		Number    &     & in ensemble            \\
    \hline
		\hline
    1 & 128  & $M=100$ \\
    2 &256   & $M=100$ \\
    3 &512   & $M=100$ \\
    4 &1024  & $M=100$ \\
    5 &2048  & $M=100$ \\
    6 &4096  & $M=100$ \\
    7 &8192  & $M=100$ \\
		\hline
\end{tabular}
\caption{Explanation of the scheme for constructing the ensemble of simulations.  We keep $M=100$ fixed throughout our investigations, however, we vary $N$ systematically to explore finite-size effects.}
\label{tab:ensembles}
\end{table}
For a fixed value of $N$, we thereby obtain the following estimates of the moments of $\beta$:
\begin{subequations}
\begin{eqnarray}
\mu_1(t)&=&\frac{1}{M}\sum_{j=1}^M \beta_j(t),\\
\mu_p(t)&=&\frac{1}{M}\sum_{j=1}^M \left[\beta_j(t)-\mu_1(t)\right]^p,\qquad p=2,3,\cdots,
\end{eqnarray}%
\label{eq:mup}%
\end{subequations}%
These moments are plotted in an interval $20<t<40$ (i.e. before the onset of extreme finite-size effects) in 
Figure~\ref{fig:moments}. 
\begin{figure}[htb]
	\centering
	\subfigure[\,\,\text{Mean}]{\includegraphics[width=0.45\textwidth]{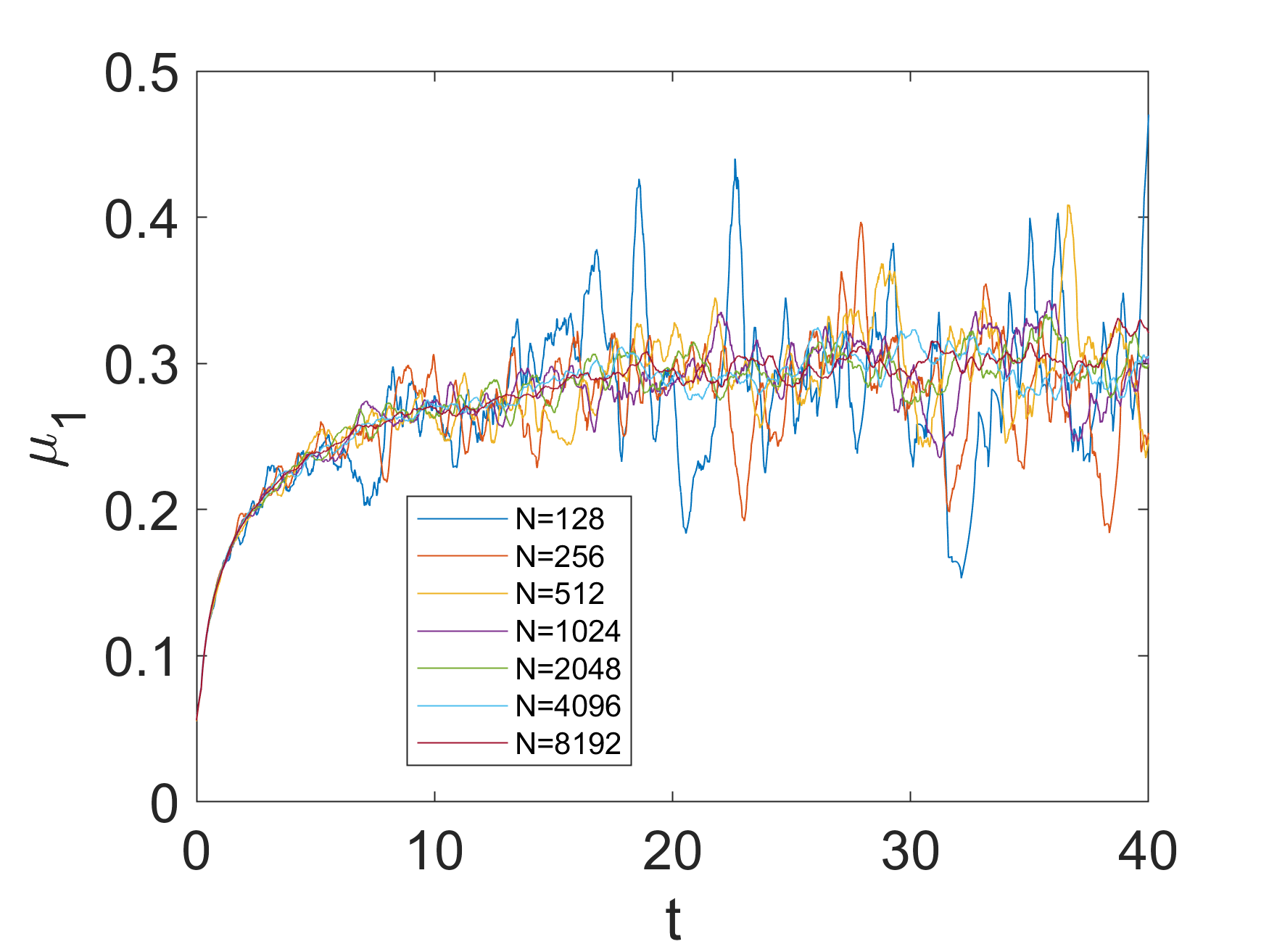}}
	\subfigure[\,\,\text{Variance}]{\includegraphics[width=0.45\textwidth]{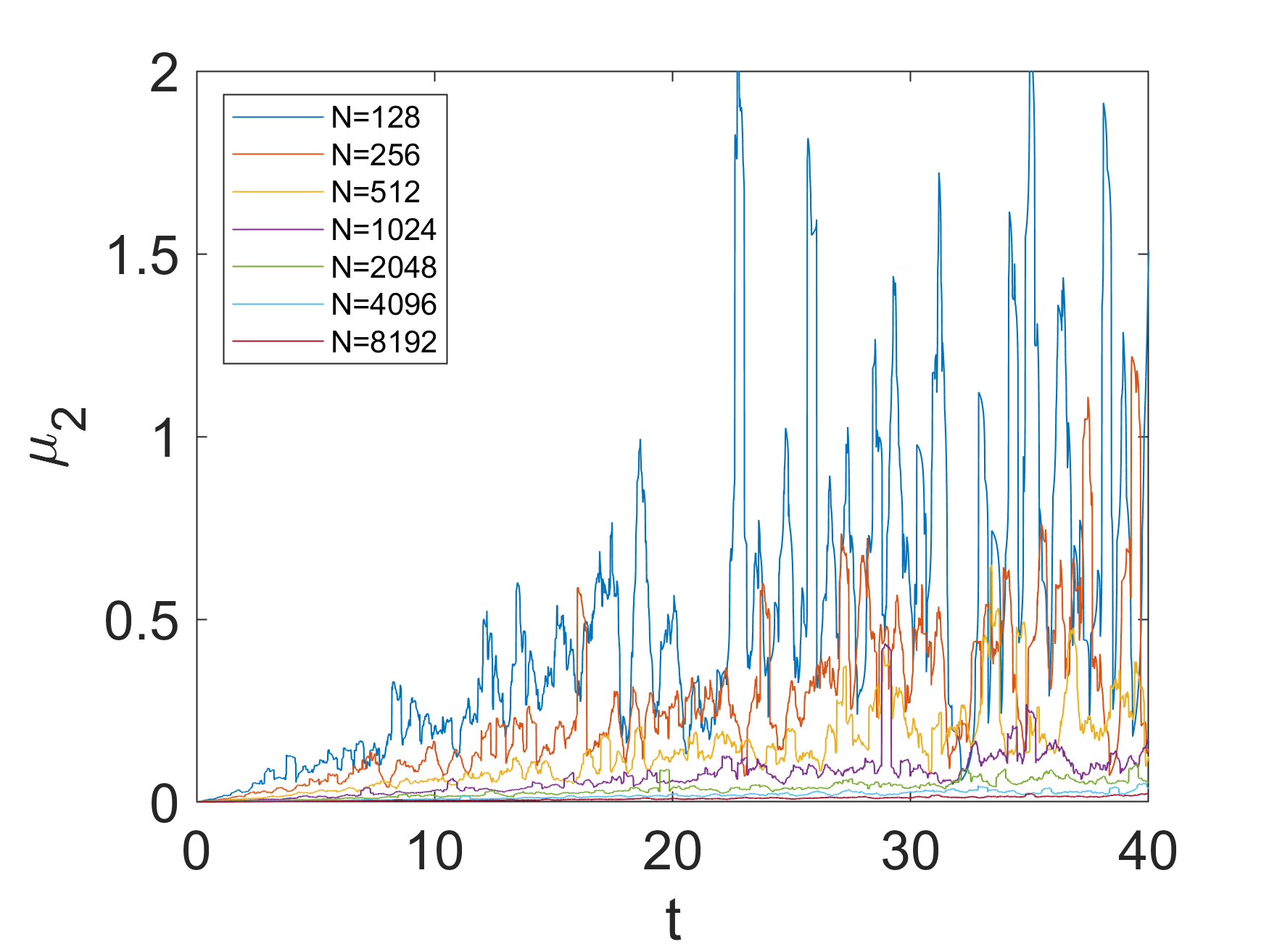}}
	\caption{Plots showing time series of  $\mu_1$ and $\mu_2$  Moving averages are shown to guide the eye.}
	\label{fig:moments}
\end{figure}
 From these results, and for $N$ sufficiently large, it can be inferred that $\mu_1$ fluctuates around a constant value $\betaconst$; the constant value can be estimated from
\begin{equation}
\betaconst\approx\frac{1}{t_2-t_1}\int_{t_1}^{t_2}\mu_1(t)\,\mathd t,\qquad t_1=20,\qquad t_2=40.
\label{eq:beta0_stats}
\end{equation}
From Figure~\ref{fig:moments}, it can be seen that a value $N\apprge 512$ is required for this description to hold.
From the same figure, a least-squares fit of $\mu_2$ may also be extracted, to reveal the trend $\mu_2\propto t$.    This is shown in more detail in Figure~\ref{fig:moments_fit}.
\begin{figure}[htb]
	\centering
	\includegraphics[width=0.6\textwidth]{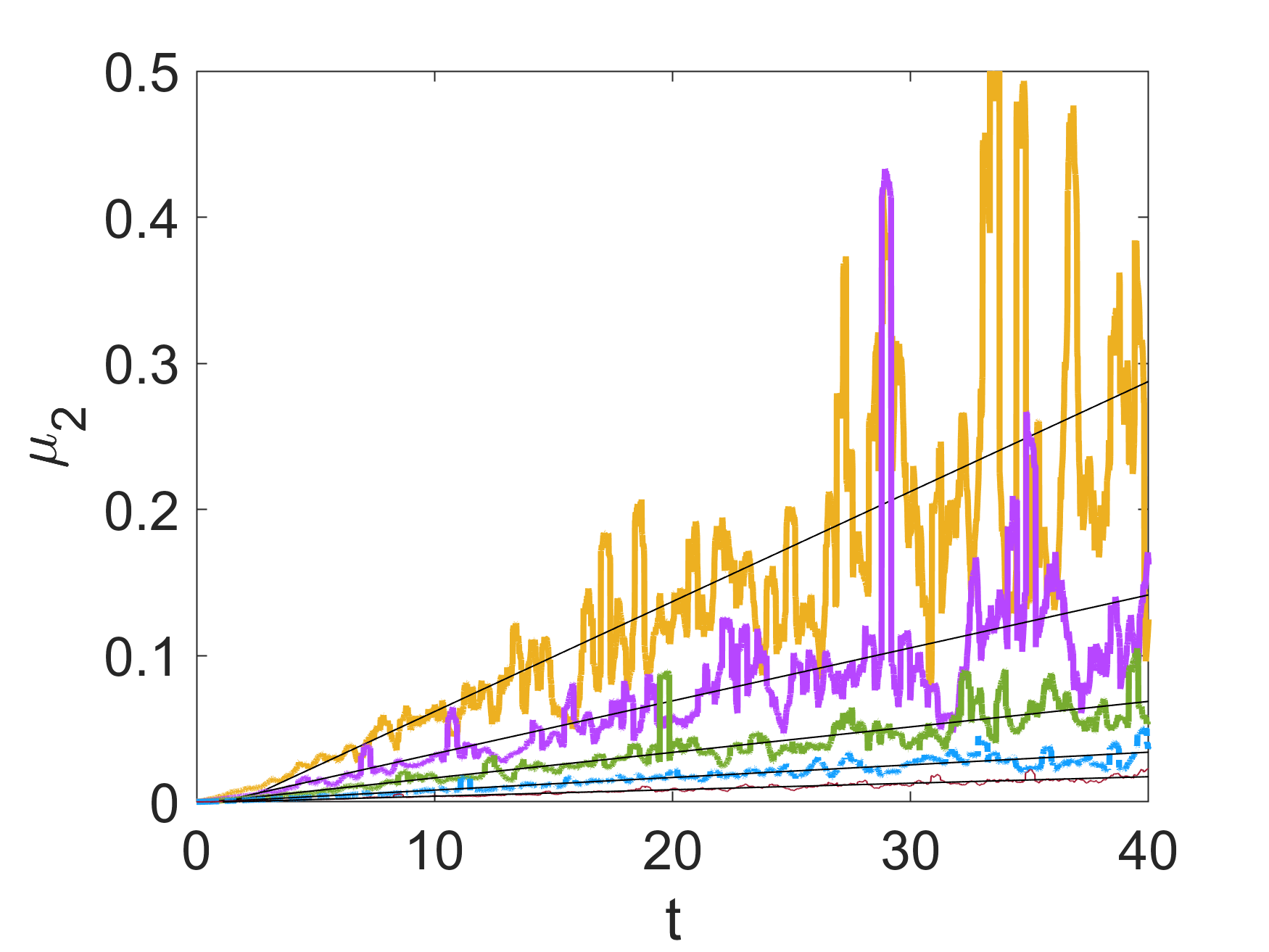}
	\caption{Plots showing time series of $\mu_2$, with least-squares fitting.  The figures shows the variance for 
	$N=512,\cdots,8192$.  The value $N=512$ corresponds to the largest variance, while the value $N=8192$ corresponds to the smallest variance.  The trend is monotone decreasing. }
	\label{fig:moments_fit}
\end{figure}
These results suggest that $\mathbb{E}(\delta\beta^2)=kt$, where $k$ is a constant.  This constant  may be estimated from the least-squares fitting.  The estimates for $k$ (and for $\betaconst$) are shown in Table~\ref{tab:stats}.  The fit for $k$ has been performed over the entire range of $t$-values, $t\in [0,40]$; the resulting values of $k$ do not change much if the range of $t$ used for the fitting is changed.
\begin{table}
\centering
\begin{tabular}{|c|c|c|} 
\hline
    {$N$} & $\frac{1}{20}\int_{20}^{40}\mu_1(t)\,\mathd t$  & $k$ \\
    \hline
		\hline
    512   & 0.24      & 0.0075  	\\
    1024  & 0.2613    & 0.0036   \\
    2048  & 0.2833    & 0.0017 	\\
    4096  & 0.2908    & $8.6\times 10^{-4}$		\\
    8192  & 0.2995    & $4.5 \times 10^{-4}$   \\
		\hline
\end{tabular}
\caption{Estimates of $\mathbb{E}(\beta)=\betaconst$ and $\mathbb{E}(\delta\beta^2)$ for Model 1, for various problem sizes $N$.}
\label{tab:stats}
\end{table}

From these results, we can infer:
\begin{eqnarray*}
\mathbb{E}\bigg\{\left[\log\left(\frac{F(t_0)t_0^{\betaconst}}{F(t)t^{\betaconst}}\right)\right]^2\bigg\} 
&\leq& (t-t_0)\int_{t_0}^{t}\frac{\mathbb{E}(\delta\beta^2)}{t^2}\mathd t,\qquad t>t_0,\\
&=&(t-t_0)\int_{t_0}^{t}\frac{k}{t}\mathd t,
\end{eqnarray*}
hence
\[
\mathbb{E}\bigg\{\left[\log\left(\frac{F(t_0)t_0^{\betaconst}}{F(t)t^{\betaconst}}\right)\right]^2\bigg\} 
\leq k (t-t_0)\log(t/t_0).
\]
i.e. Equation~\eqref{eq:beta_model3} in the introduction.  Moreover, the trend towards the classical scaling behaviour $\betaconst=1/3$ and 
\[
\mathbb{E}\bigg\{\left[\log\left(\frac{F(t_0)t_0^{\betaconst}}{F(t)t^{\betaconst}}\right)\right]^2\bigg\} =0
\]
is in evidence as $N$ increases.  In particular, the $k$-value halves with a doubling of $N$, suggestive of $k\sim N^{-1}$, and hence $[F(t)t^{1/3}]/[F(t_0)t_0^{1/3}]=1$ as $N\rightarrow\infty$.  This is consistent with the conjectured bound in Equation~\eqref{eq:kohn_maybe}.

We conclude this section by looking at the probability distribution function of $\beta$, thus addressing {\textbf{Question 2}} in the introduction.  Since $\mu_2$ depends on time, the distribution of $\beta$ is not stationary -- this is reinforced by the fact that the higher moments  $\mu_4$ (not shown) also has a systematic variation with time.  As such, it is appropriate only to plot a space-time evolution  of the histogram of $\beta$ -- this is shown in Figure~\ref{fig:hist8192_beta}.
\begin{figure}[htb]
    \centering
     \includegraphics[width=0.6\textwidth]{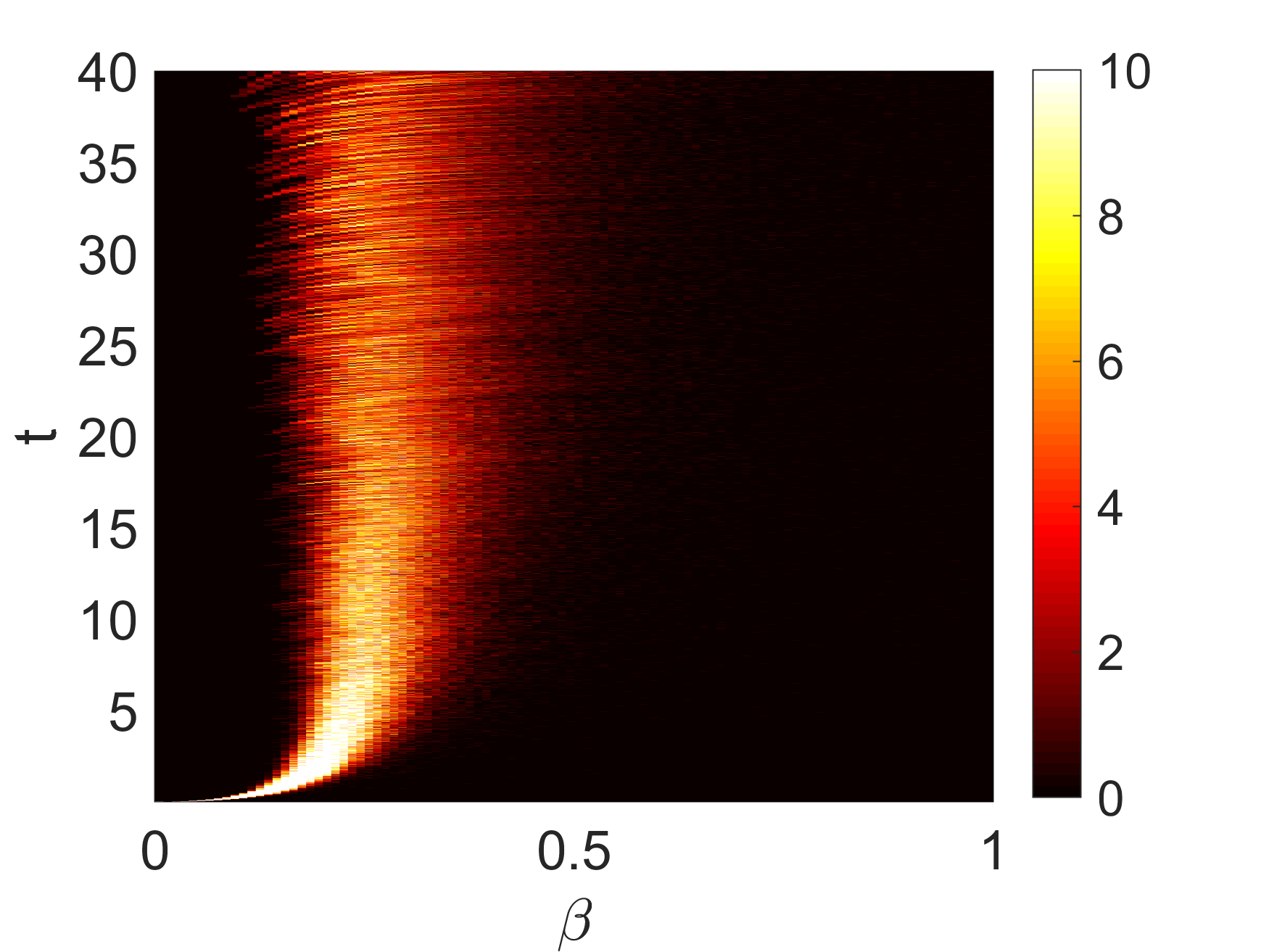}
				\caption{Spacetime plot of the histogram of $\beta$ for the case $N=8192$.   }
    \label{fig:hist8192_beta}
\end{figure}
However, since the variance $\mu_2$ associated with the histograms in Figure~\ref{fig:hist8192_beta} decreases with increasing $N$ (Table~\ref{tab:stats}), it can be inferred that the histogram of $\beta$ does approach a delta function as $N\rightarrow\infty$ -- only the approach is not self-similar, such that the histogram retains a time-dependent form for all finite values of $N$.

\section{Model 2 -- Numerical Simulations}
\label{sec:CH}

In this section we solve Model 2  numerically in an ensemble of $M$ simulations.  We thereby build up a statistical picture of $\beta$.  We recall Model 2 as the Cahn--Hilliard equation~\eqref{eq:ch}.  We solve Model 2 in two spatial dimensions, using periodic boundary conditions.  As such, the following mean concentration is conserved:
\begin{equation}
\langle C\rangle=\frac{1}{|\Omega|}\int_{\Omega}C(\vecx,t)\mathd^2 x.
\end{equation}
In this section, we start by looking at asymmetric mixtures, corresponding to $\langle C\rangle \neq 0$.  This corresponds closely to LSW theory and hence, to Model 1.  Thereafter, we look also at symmetric mixtures with $\langle C\rangle\neq 0$.  The morphology of symmetric mixtures does not fit into the LSW framework, however, the scaling hypothesis in Equation~\eqref{eq:kohn_maybe} applies equally well to either scenario, therefore, it is worthwhile to study both.

We emphasize that although the droplet dynamics, LSW theory, and the underlying connection to the Cahn--Hilliard equation (Sections~\ref{sec:intro}--\ref{sec:drop_pop}) are defined for $D=3$ dimensions only, there is an analogous quantitative theory for $D=2$ (see Reference~\cite{rogers1989numerical}).  The outcome of that quantitative theory is again a self-similar dropsize distribution whose form is similar to that already explored in Section~\ref{sec:drop_pop}.  Therefore, the results of Section~\ref{sec:drop_pop} (for $D=3$) will carry over in a suitable qualitative sense to asymmetric mixtures in the present section (for $D=2$).  

\subsection{Methodology}

The Cahn--Hilliard equation~\eqref{eq:ch} is solved numerically using a semi-implicit Alternating Direction Implicit (ADI) finite-difference method, with periodic boundary conditions in both spatial directions.    The computational grid has $n$ gridpoints in each spatial direction, and the time-step is denoted by $\Delta t$.
The size of the physical domain is $L$, such that $|\Omega|=L^2$.
The numerical algorithm is implemented in the CUDA programming language, where each direction of the ADI scheme is parallelised using the methodology presented in References~\cite{gloster2019cupentbatch,gloster2019custen}.  The resulting computer code is implemented on Graphics Processing Units (GPUs) -- specifically, 
two NVIDIA Titan X GPUs with 12GB of RAM each and an Intel i7-6850K CPU with 6 3.60GHz hyper–threaded cores; the host machine is equipped with 128GB of RAM.  The operating system is Ubuntu 18.04 LTS and we use CUDA v10.1.  
The GPU uses the CUDA MPS server to solve a large ensemble of simulations ($M=1024$) in batches on the GPUs.  
The numerical method (with validations) is described in detail in References~\cite{gloster2019cupentbatch,gloster2019custen}.

In the present article, we use the use the value $\gamma= 0.01$ throughout.  As such, numerical convergence of the method is achieved with a spatial resolution of $\Delta x=2\pi/256$, and a timestep of $\Delta t=0.0012$.  This resolution is kept constant throughout the course of our investigations, while the domain size $L$ is systematically varied.  Quantitative evidence of numerical convergence with these parameters is given at a key point towards the end of this section, while the numerical convergence is further investigated in a systematic way in Appendix~\ref{sec:app:cvg}.  
Finally, the initial condition is set as
\[
C(\vecx,t=0)=0.5+0.1[r(\vecx)-1],\qquad \text{Asymmetric Mixture},
\]
and
\[
C(\vecx,t=0)=r(\vecx),\qquad \text{Symmetric Mixtures}.
\]
Here, $r(\vecx)$ is a random number drawn from the uniform random distribution, with $0\leq r(\vecx)\leq 1$; the random numbers  at the different points $\vecx\in\Omega$ are all drawn from independent identical distributions. 

\subsection{Results -- Asymmetric Mixture}

Snapshots of $C(\vecx,t)$ at different times are shown in Figure~\ref{fig:oswaldRipening} for a typical simulation of an asymmetric mixture ($L=2\pi$).  It can be seen in this illustrative example how the system evolves: $C(\vecx,t)$ rapidly evolves to a configuration of many small droplets, these undergo Ostwald ripening such that only a few large droplets survive;  eventually the system will consist of only a single droplet, corresponding to the extreme finite size effects already identified in Section~\ref{sec:drop_pop}.
\begin{figure}
\centering
\subfigure[$\,\,t \approx  6.14$]{\includegraphics[width=0.22\textwidth]{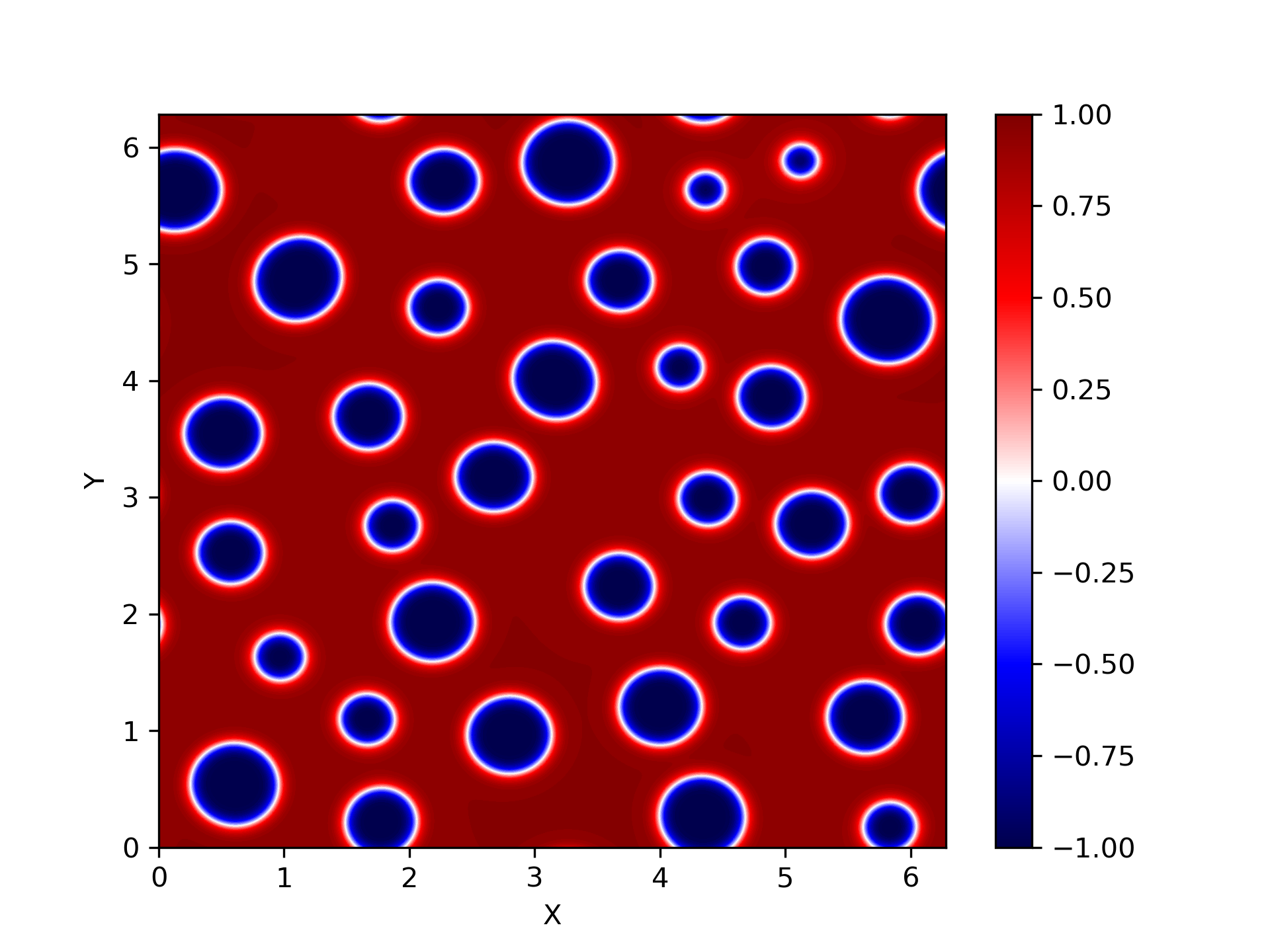}}
\subfigure[$\,\,t \approx 12.52$]{\includegraphics[width=0.22\textwidth]{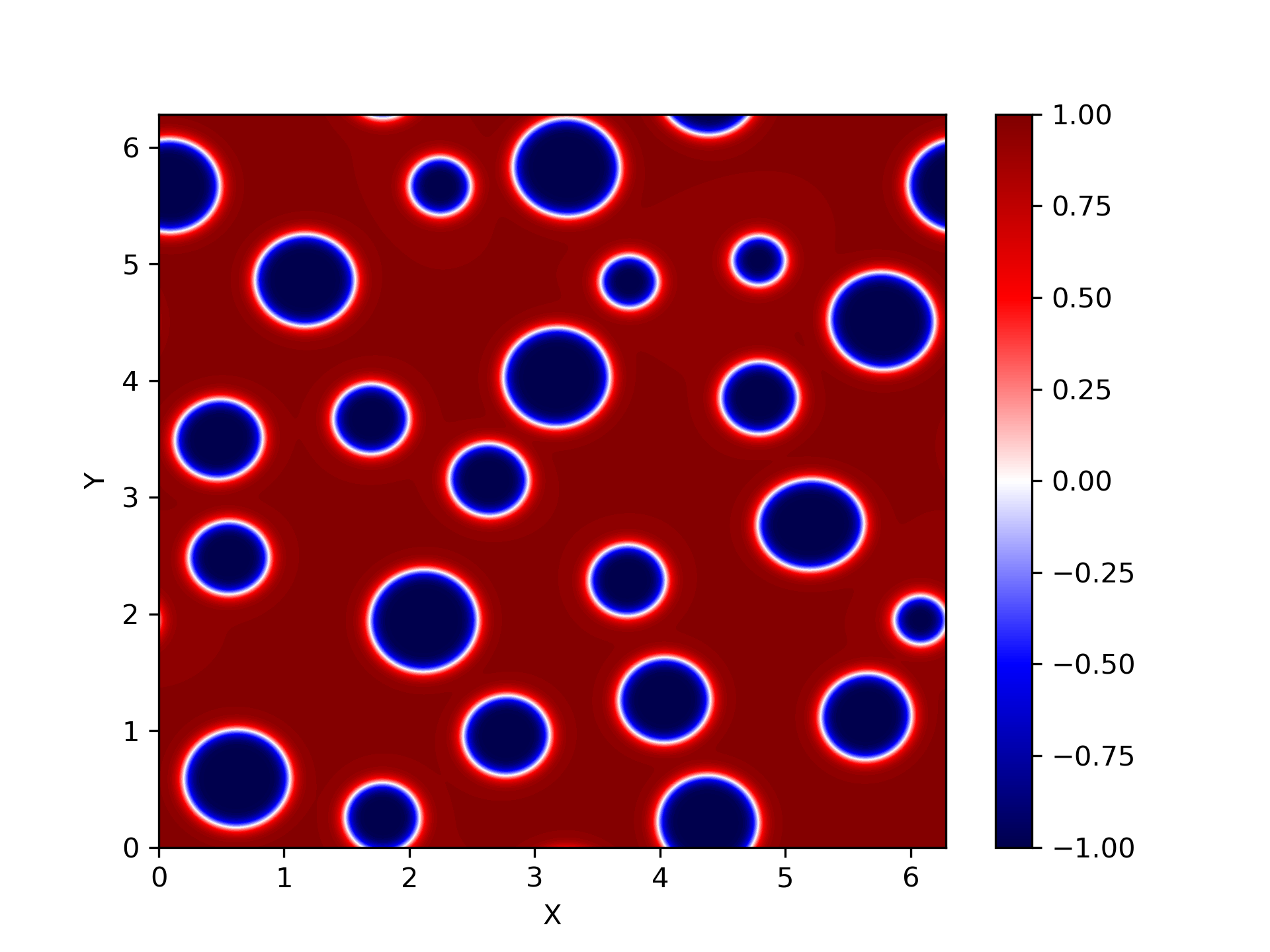}}
\subfigure[$\,\,t \approx 24.79$]{\includegraphics[width=0.22\textwidth]{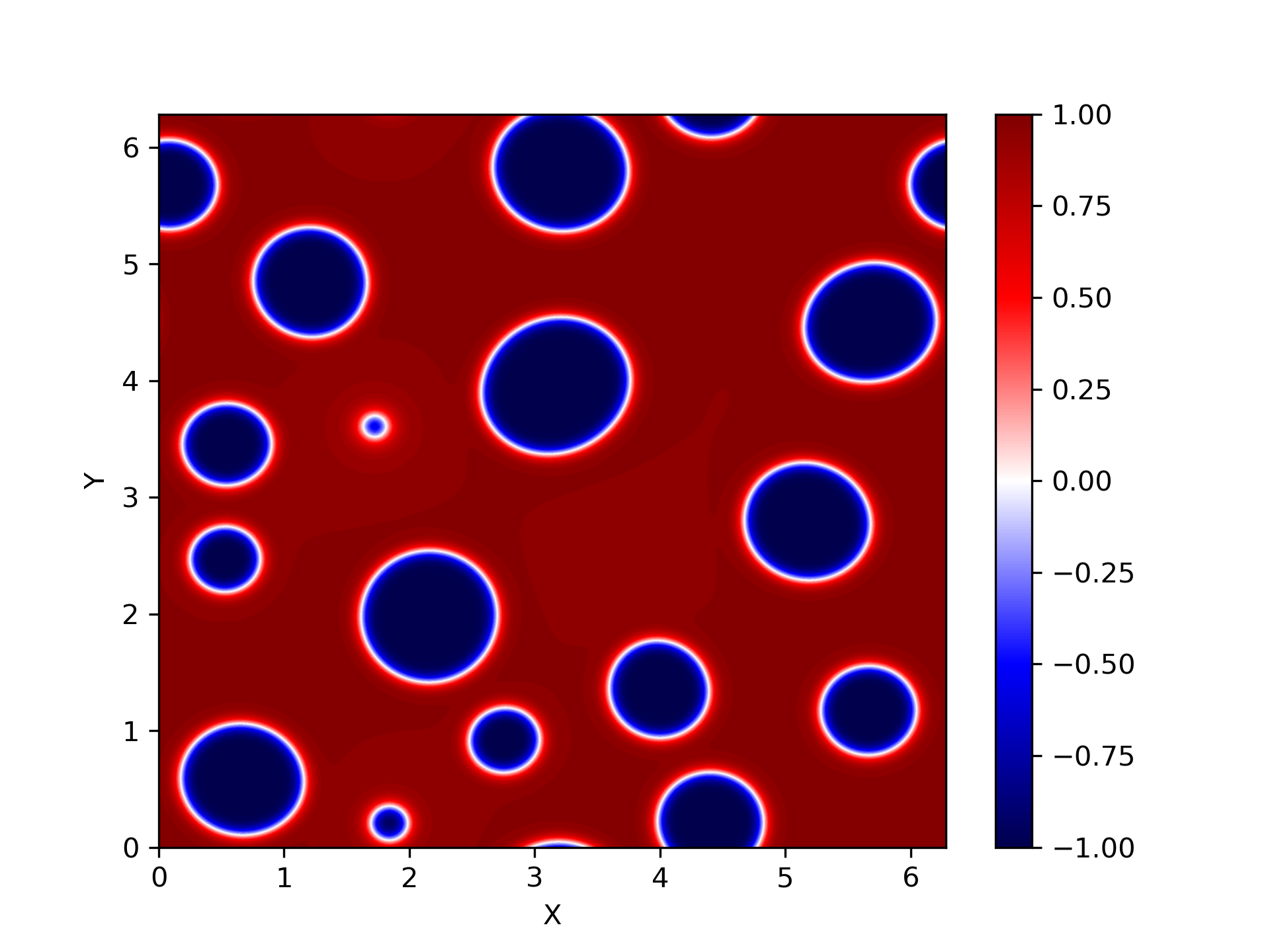}}
\subfigure[$\,\,t \approx 49.33$]{\includegraphics[width=0.22\textwidth]{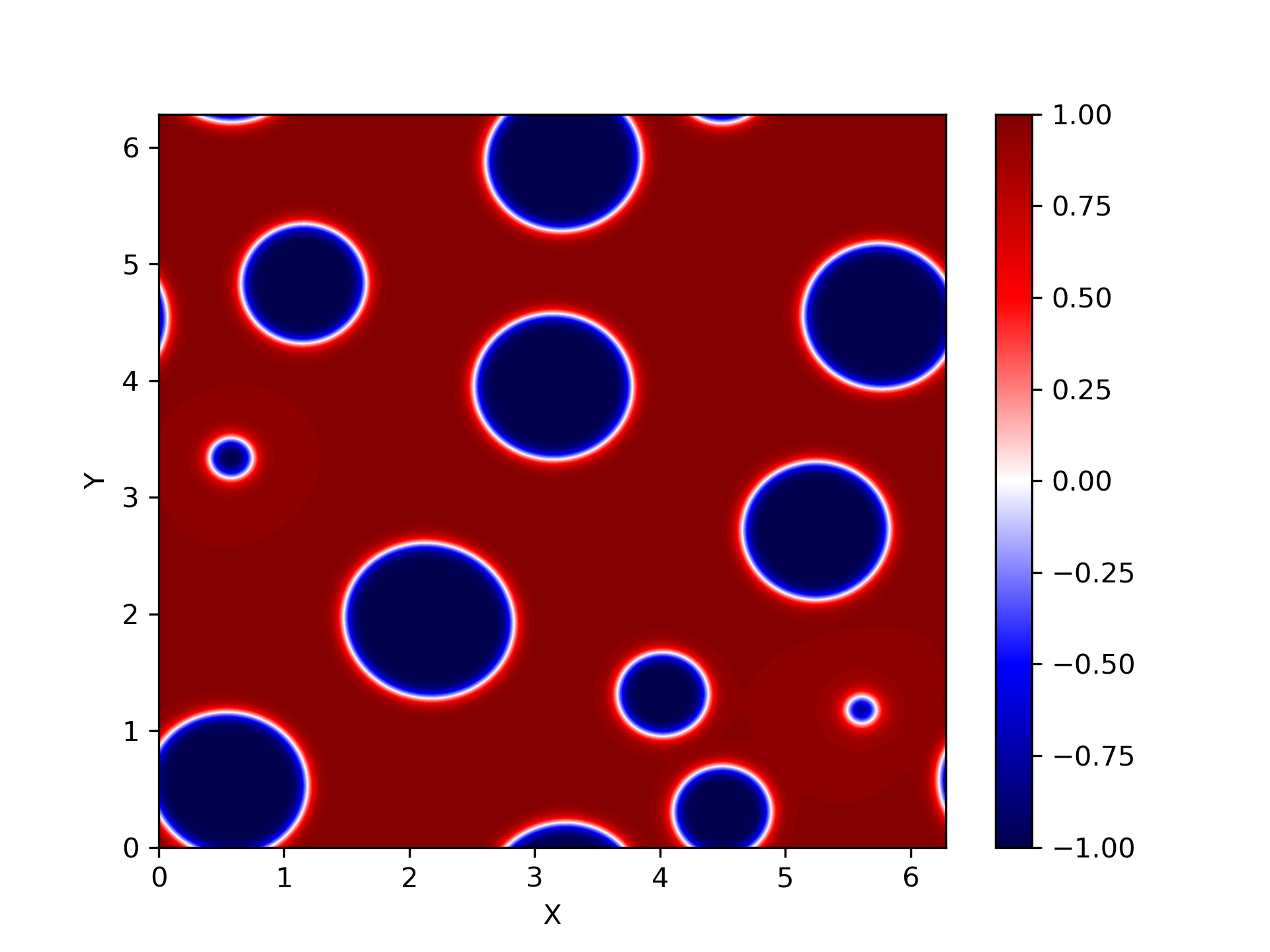}}
\caption{Contour plots of the Cahn--Hilliard equation with asymmetric mixture showing Oswald Ripening on a domain of size $2\pi \times 2\pi$.}
\label{fig:oswaldRipening}
\end{figure}
A time series of $\beta$ for a single simulation ($L=16\pi$) is shown in Figure~\ref{fig:2Dbeta}.  
\begin{figure}
	\centering
		\includegraphics[width=0.6\textwidth]{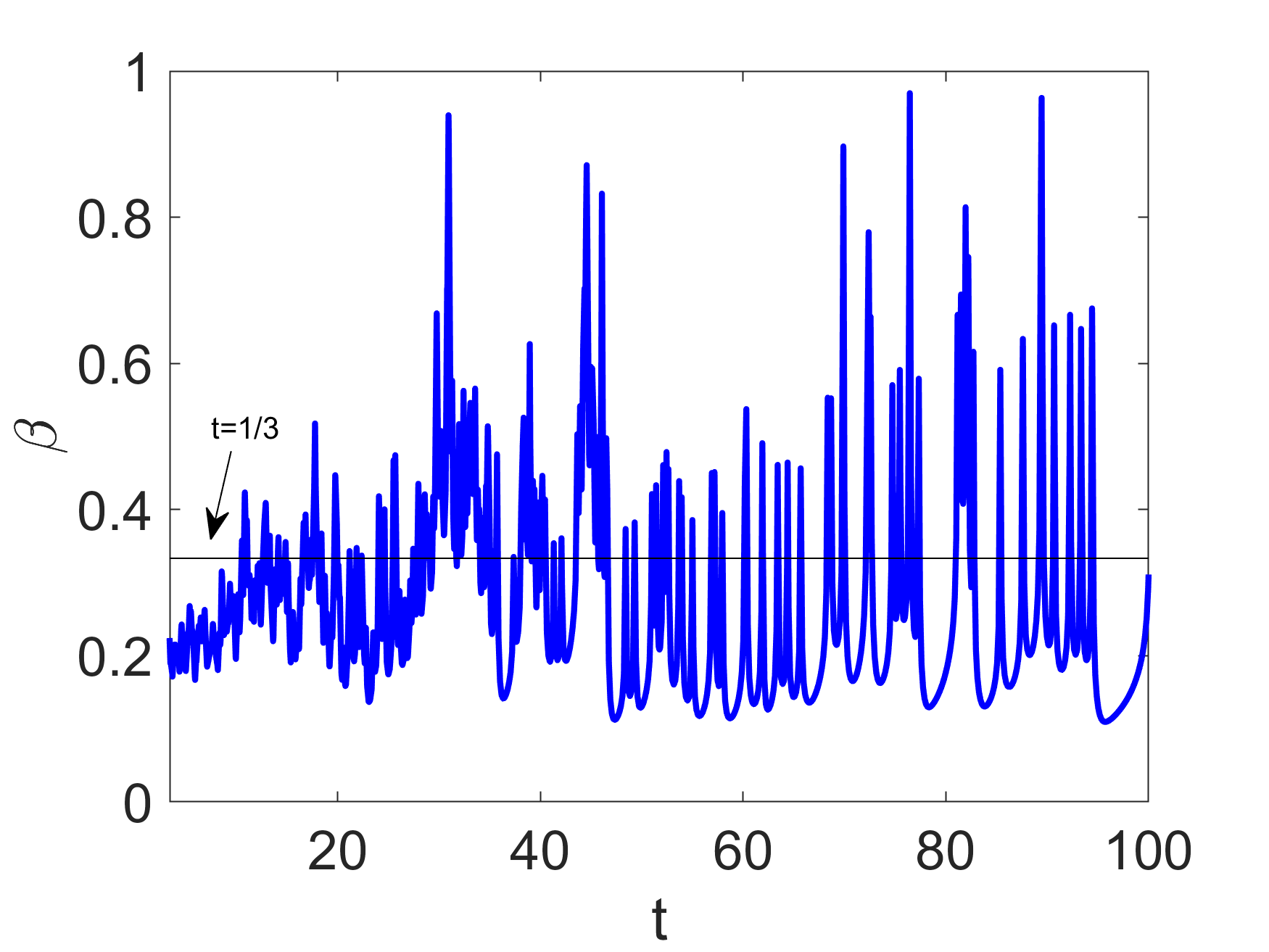}
		\caption{Time series of $\beta(t)$ for a single typical simulation of the Cahn--Hilliard equation (asymmetric mixture, $L=16\pi$).  The line $\beta(t)=1/3$ is added  for reference.}
	\label{fig:2Dbeta}
\end{figure}
This is seen to be qualitatively similar to the corresponding time series for Model 1 (e.g. Section~\ref{sec:drop_pop}, Figure~\ref{fig:beta_N10000}).

We perform a campaign of numerical simulations to characterize the coarsening rate $\beta$ in the Cahn--Hilliard equation.  
As $\beta=\betaconst+\delta\beta$ is a characterization of an entire droplet population (i.e. a single simulation of the Cahn--Hilliard equation in a domain $\Omega$), it is necessary to gather statistics of $\beta$ across ensembles of many such simulations -- see Table~\ref{tab:ensemblesCH}.
\begin{table}[htb]
\centering
\begin{tabular}{|c|c|c|} 
\hline
    Ensemble  & $|\Omega|=L^2$ & Number of simulations  \\
		Number    & $L$    &         in ensemble            \\
    \hline
		\hline
    1 & $2\pi$  & $M=1024$ \\
    2 & $4\pi$   & $M=1024$ \\
    3 & $8\pi$   & $M=1024$ \\
    4 & $16\pi$  & $M=1024$ \\
		\hline
\end{tabular}
\caption{Explanation of the scheme for constructing the ensemble of simulations.  We keep $M=1024$ fixed throughout our investigations, however, we vary $L$ systematically to explore finite-size effects.  The timestep is kept fixed throughout as $\Delta t=0.0012$.  Also, the spatial resolution is kept fixed throughout as $\Delta x=2\pi/512$.  Hence, the number of gridpoints in a simulation of size $L$ is $n=(L/\Delta x)$, with $n=\{512,1024,2048,4096\}$.}
\label{tab:ensemblesCH}
\end{table}
The statistical moments of the ensemble data are generated -- again in accordance to Equation~\eqref{eq:mup}.  The results of this analysis for $\mu_1$ and $\mu_2$ are shown in time-series form in Figure~\ref{fig:moments_CH}.
\begin{figure}[htb]
	\centering
	\subfigure[\,\,\text{Mean}]{\includegraphics[width=0.45\textwidth]{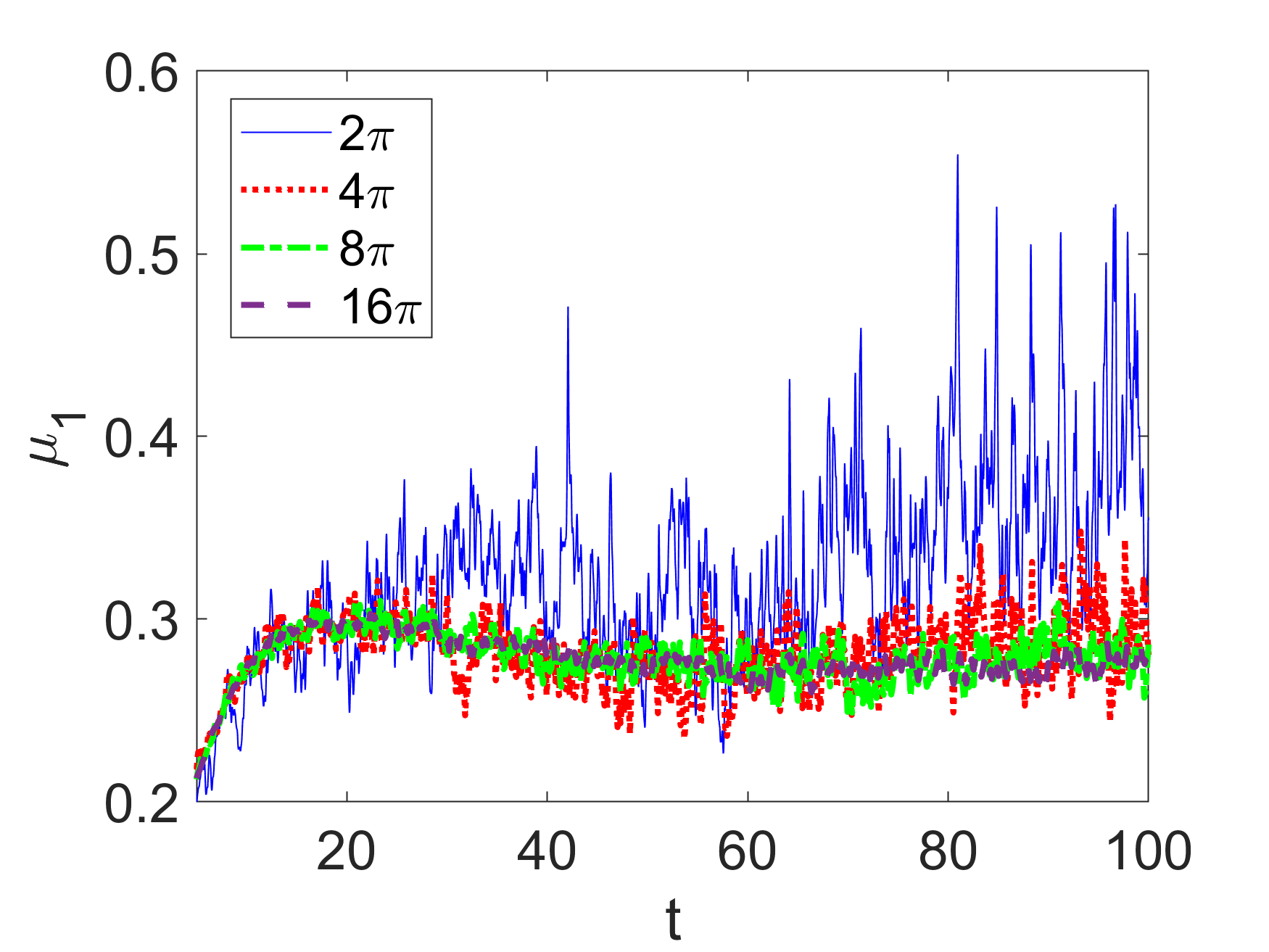}}
	\subfigure[\,\,\text{Variance}]{\includegraphics[width=0.45\textwidth]{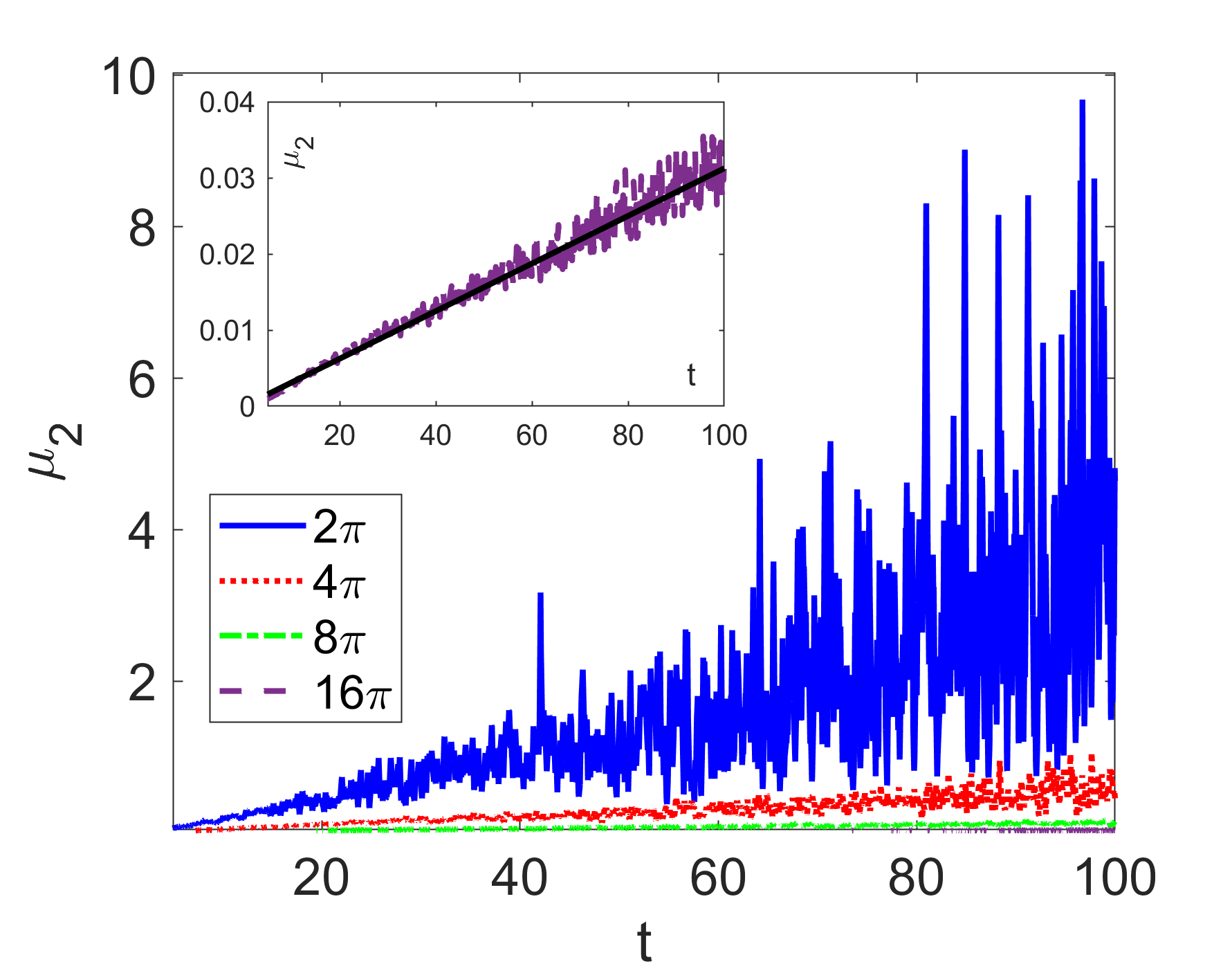}}
	\caption{Plots showing time series of  $\mu_1$ and $\mu_2$ for asymmetric mixtures.  Moving averages are shown to guide the eye.  The inset in (b) shows the $L=16\pi$ case on an enlarged scale, with a least-squares fit $\mu_2(t)=kt$ to highlight the systematic drift in $\mu_2(t)$, and to guide the eye.}
	\label{fig:moments_CH}
\end{figure}
The trends are consistent with what was observed already in in Model 1 -- specifically, $\beta(t)$ fluctuates around a constant value $\beta_0$, and the amplitude of the fluctuations increases over time.  The amplitude of the fluctuations is quantified by $\mathbb{E}(\delta\beta^2)$, where $\delta\beta=\beta(t)-\beta_0$.  The amplitude of the fluctuations is seen to increase linearly in time (e.g. inset to Figure~\ref{fig:moments_CH}).  This may be modelled as $\mathbb{E}(\delta\beta^2)\propto t$; the constant of proportionality may be estimated by least-squares fitting of $\mu_2(t)$ to a trend-line $\mu_2=kt$.  The results of applying this statistical model to the data are summarized in Table~\ref{tab:stats_CH}; this is again consistent with the statistical description of Model 1 in Section~\ref{sec:drop_pop}.
\begin{table}
\centering
\begin{tabular}{|c|c|c|} 
\hline
    {$L$} & $\frac{1}{80}\int_{20}^{100}\mu_1(t)\,\mathd t$ & $k$ \\
    \hline
		\hline
    $2\pi$  & 0.33      & 0.036  	\\
    $4\pi$  & 0.28      & 0.0056   \\
    $8\pi$  & 0.28      & 0.0013 	\\
    $16\pi$ & 0.28      & 0.00031		\\
		\hline
\end{tabular}
\caption{Estimates of $\mathbb{E}(\beta)=\betaconst$ and $\mathbb{E}(\delta\beta^2)$ for Model 2 (asymmetric mixtures), for various problem sizes $L$.  The value of $k$ is obtained from least-squares fitting on the data between $t=20$ and $t=100$.}
\label{tab:stats_CH}
\end{table}

From Table~\ref{tab:stats_CH}, it can be seen that the value of $k$ is domain-dependent; it can be further seen that $k\sim L^{-2}$ (see also Figure~\ref{fig:ksimL}).
\begin{figure}[htb]
	\centering
		\includegraphics[width=0.6\textwidth]{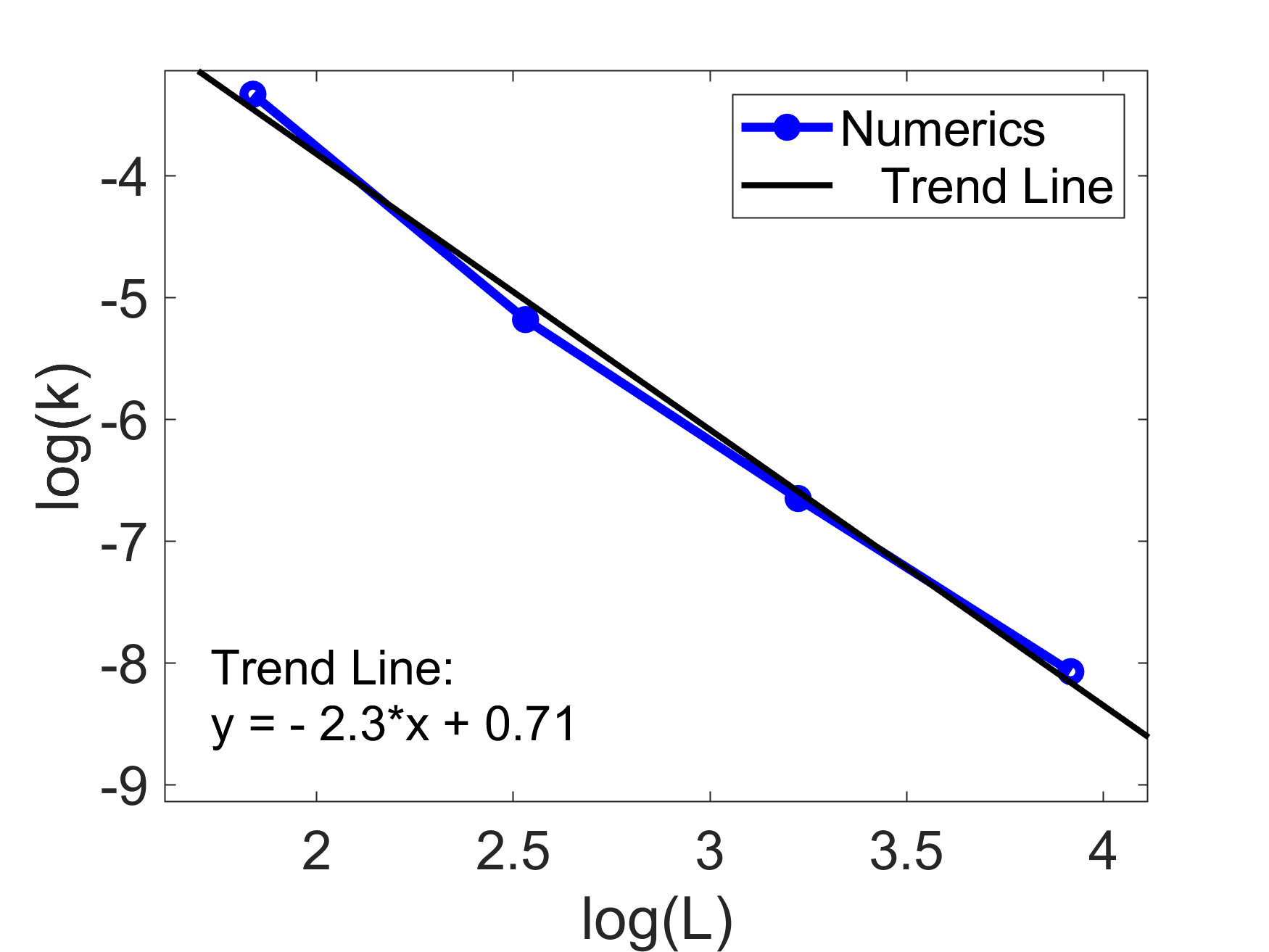}
		\caption{Plot showing the dependence of the value of $k$ on the problem size $L$.  The values of $k$ are obtained directly from Table~\ref{tab:stats_CH}.}
	\label{fig:ksimL}
\end{figure}
Hence, referring back to the arguments in the introduction, the numerical evidence suggests that in the limit as $L\rightarrow\infty$,  $[F(t)t^{\beta_0}]/[F(t_0)t_0^{\beta_0}]=1$.  In contrast to Model 1 (Section~\ref{sec:drop_pop}), the asymptotic value of $\beta_0$ implied by Table~\ref{tab:stats_CH} is strictly less than $1/3$.  We emphasize that this is  consistent with the hypothesis in Equation~\eqref{eq:kohn_maybe}.

We conclude this section by looking at the probability distribution function of $\beta$ for asymmetric mixtures and the Cahn--Hilliard equation, thus addressing {\textbf{Question 3}} in the introduction. 
As in the droplet population model (Section~\ref{sec:drop_pop}), the distribution of $\beta$ is not stationary -- this is 
evidenced by the fact that $\mu_2$ increases linearly with time (Table~\ref{tab:stats_CH}).  The same trend can also be seen in $\mu_4$ (not shown).
As such, it is appropriate only to plot a space-time evolution  of the histogram of $\beta$ -- this is shown in Figure~\ref{fig:hist_asymmetric}, for the case $L=16\pi$.
The drift in the variance $\mu_2$ is evidenced by the increase in the positive tail of the histogram at late times.  
However, since the variance $\mu_2$ associated with the histograms in Figure~\ref{fig:hist8192_beta} decreases as $L^{-2}$ (Figure~\ref{fig:ksimL}), it can be inferred that the histogram of $\beta$ does approach a delta function as $L\rightarrow\infty$ -- only the approach is not self-similar, such that the histogram retains a time-dependent form for all finite values of $L$. 
This is a very similar scenario to the one on show in the discrete droplet population model in Section~\ref{sec:drop_pop}.
\begin{figure}[H]
    \centering
     \includegraphics[width=0.6\textwidth]{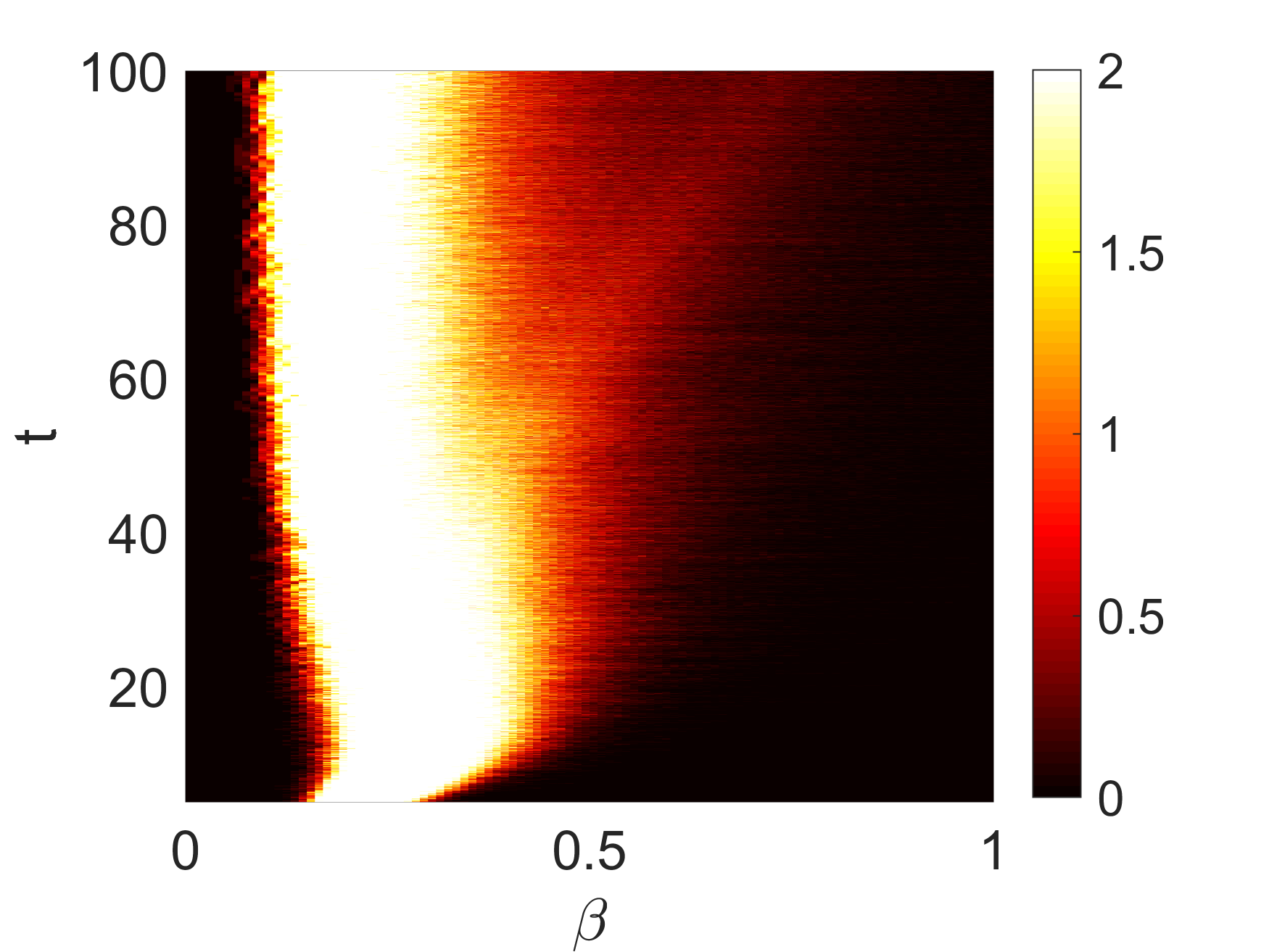}
				\caption{Spacetime plot of the histogram of $\beta$ for asymmetric mixtures and the Cahn--Hilliard equation, for the case $L=16\pi$.     }
    \label{fig:hist_asymmetric}
\end{figure}

\subsection{Results -- Symmetric Mixtures}

\begin{figure}[htb]
\centering
\subfigure[$\,\,t\approx 0.61$]{\includegraphics[width=0.22\textwidth]{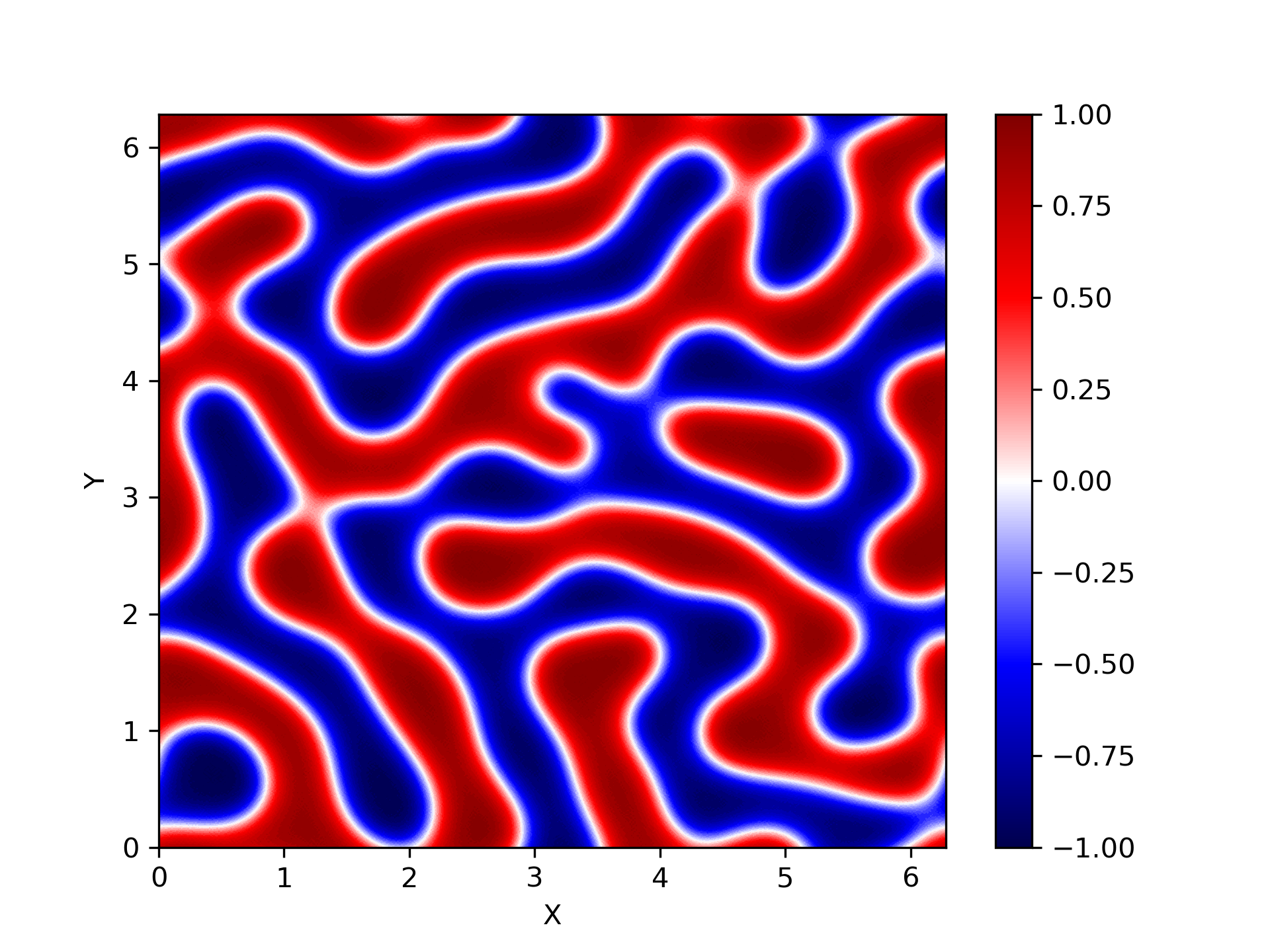}}
\subfigure[$\,\,t\approx 3.80$]{\includegraphics[width=0.22\textwidth]{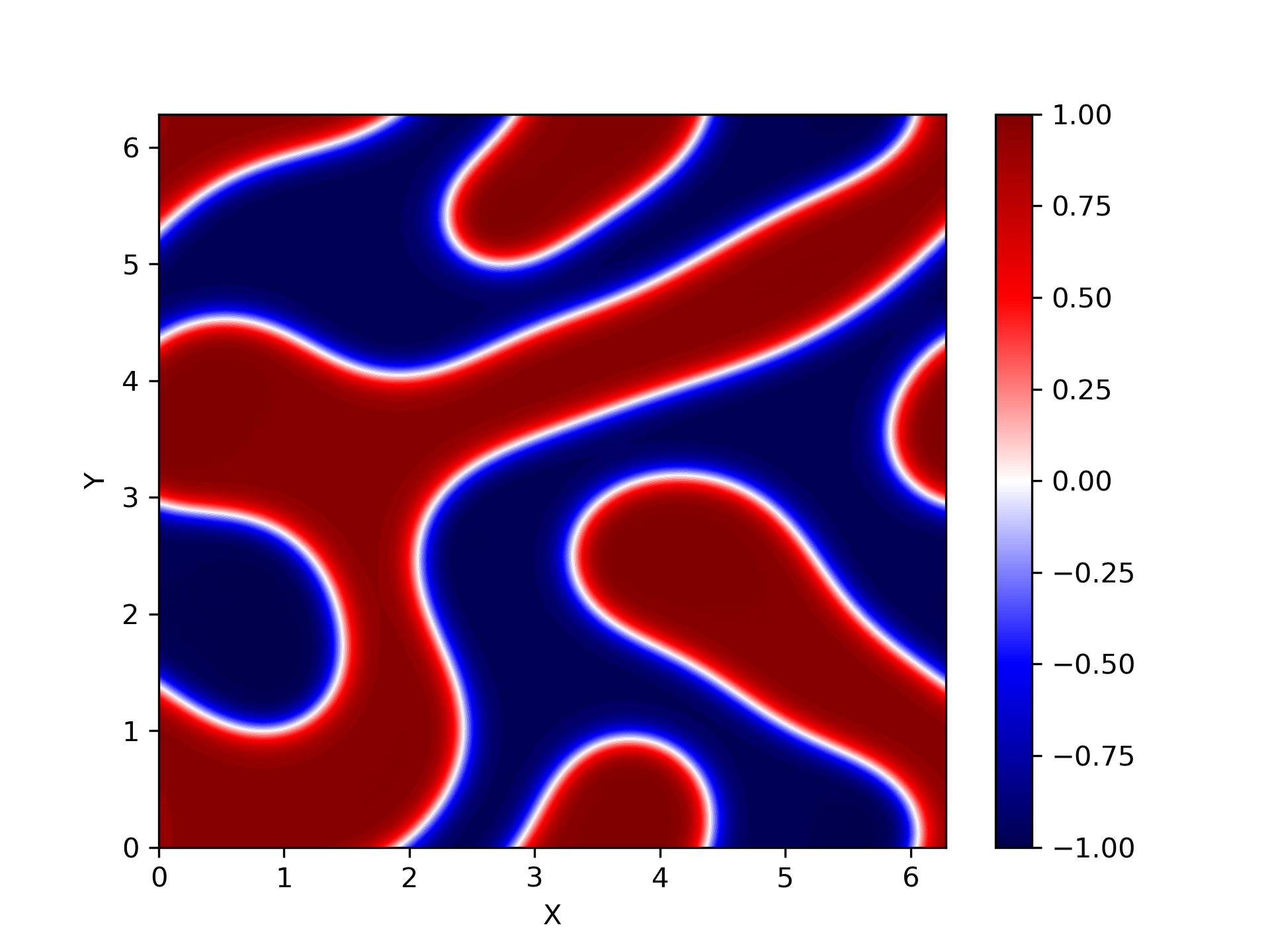}}
\subfigure[$\,\,t\approx 14.8$]{\includegraphics[width=0.22\textwidth]{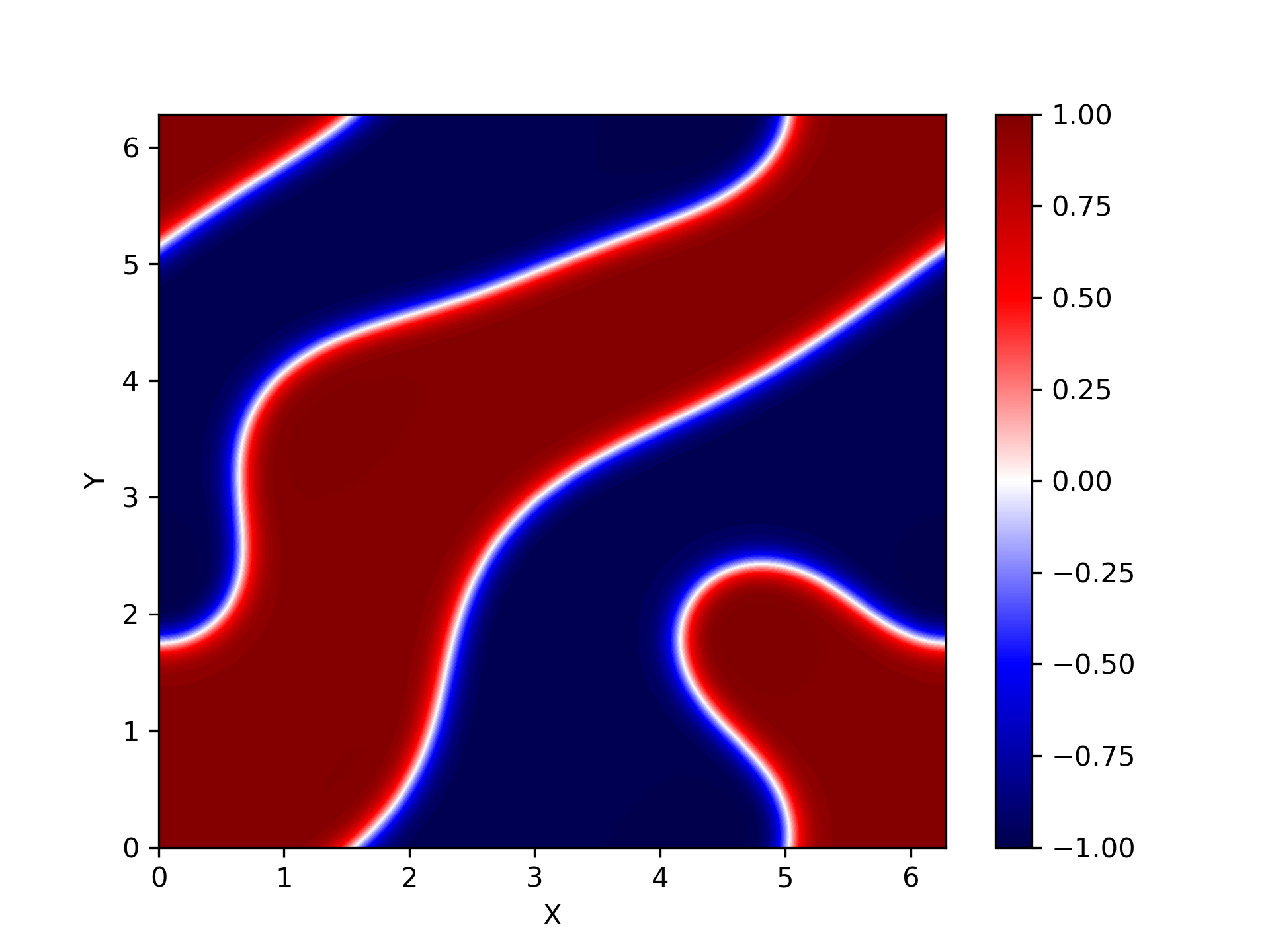}}
\subfigure[$\,\,t\approx 61.5$]{\includegraphics[width=0.22\textwidth]{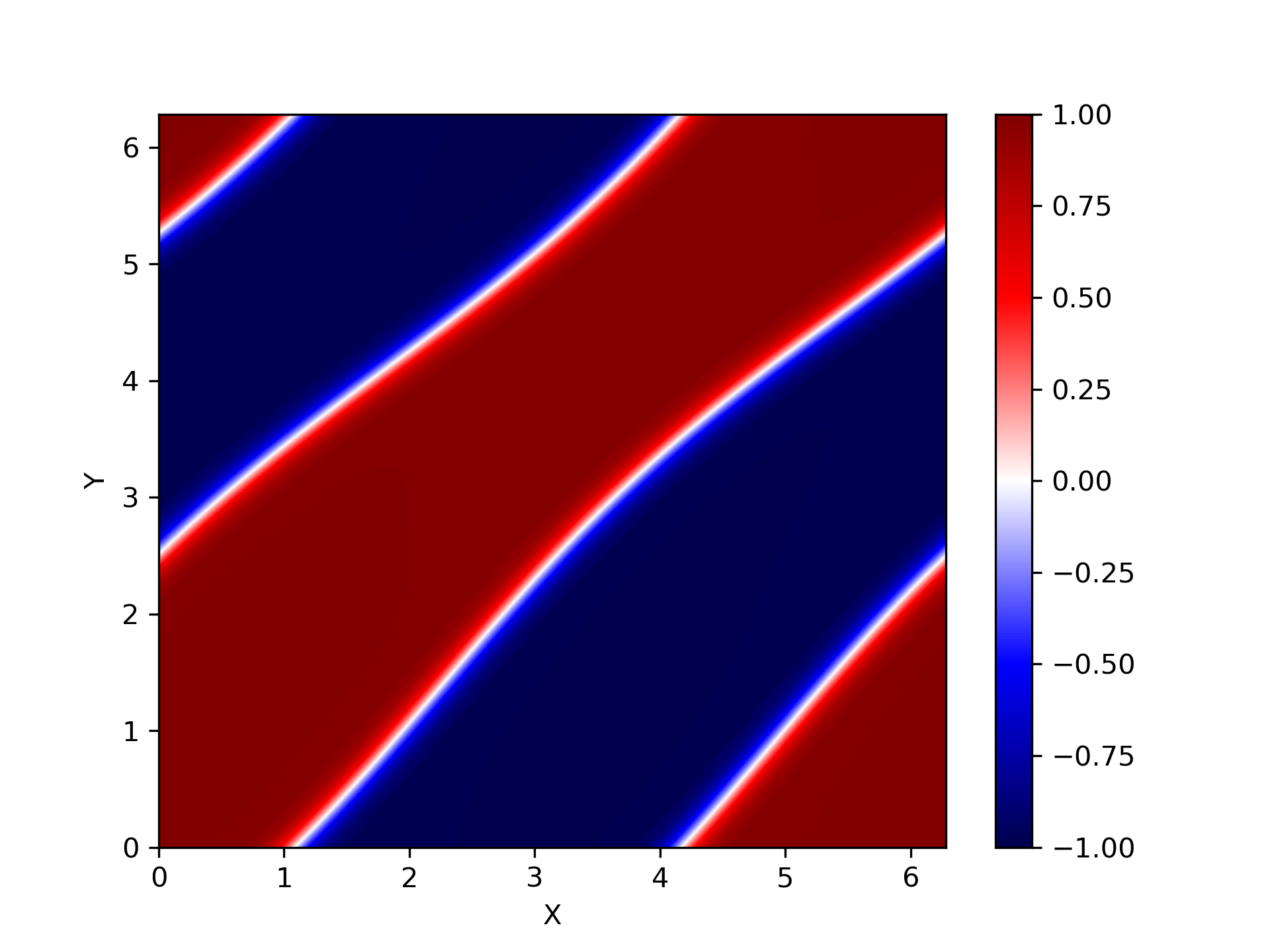}}
\caption{Sample snapshots of the Cahn--Hilliard equation for a symmetric mixture showing interconnected regions on a domain of size $2\pi \times 2\pi$. Finite-size features are apparent in two lower panels.}
\label{fig:contourSymmetric}
\end{figure}

We next look at symmetric mixtures, corresponding to $\langle C\rangle=0$. 
Snapshots of $C(\vecx,t)$ at different times are shown in Figure~\ref{fig:contourSymmetric} for a typical simulation of a symmetric mixture ($L=2\pi$).   The figure shows that $C(\vecx,t)$ rapidly evolves to form an interconnected domain structure: this is greatly in contrast to the symmetric case, where the minority phase rapidly forms droplets embedded in the majority phase.  Extreme finite-size effects are in evidence in late time, where the domain structures extend across the length of the container volume.  The onset of such extreme finite-size effects can be delayed by increasing the domain size beyond $L=2\pi$.  As such,
a time series of $\beta$ for a single simulation ($L=16\pi$) is shown in Figure~\ref{fig:2Dbeta_symmetric}.  
At late times (before the onset of extreme finite-size effects), the trend is for $\beta$ to remain constant at a value $\beta<1/3$ for long intervals, followed by sharp jumps where $\beta$ exhibits a `spike' -- these are found in the simulations to correspond to domain death (where small domains disappear via an Ostwald-ripening-type process, only to be reabsorbed into larger domains), as well as to domain merger events. 

\begin{figure}[htb]
	\centering
		\includegraphics[width=0.6\textwidth]{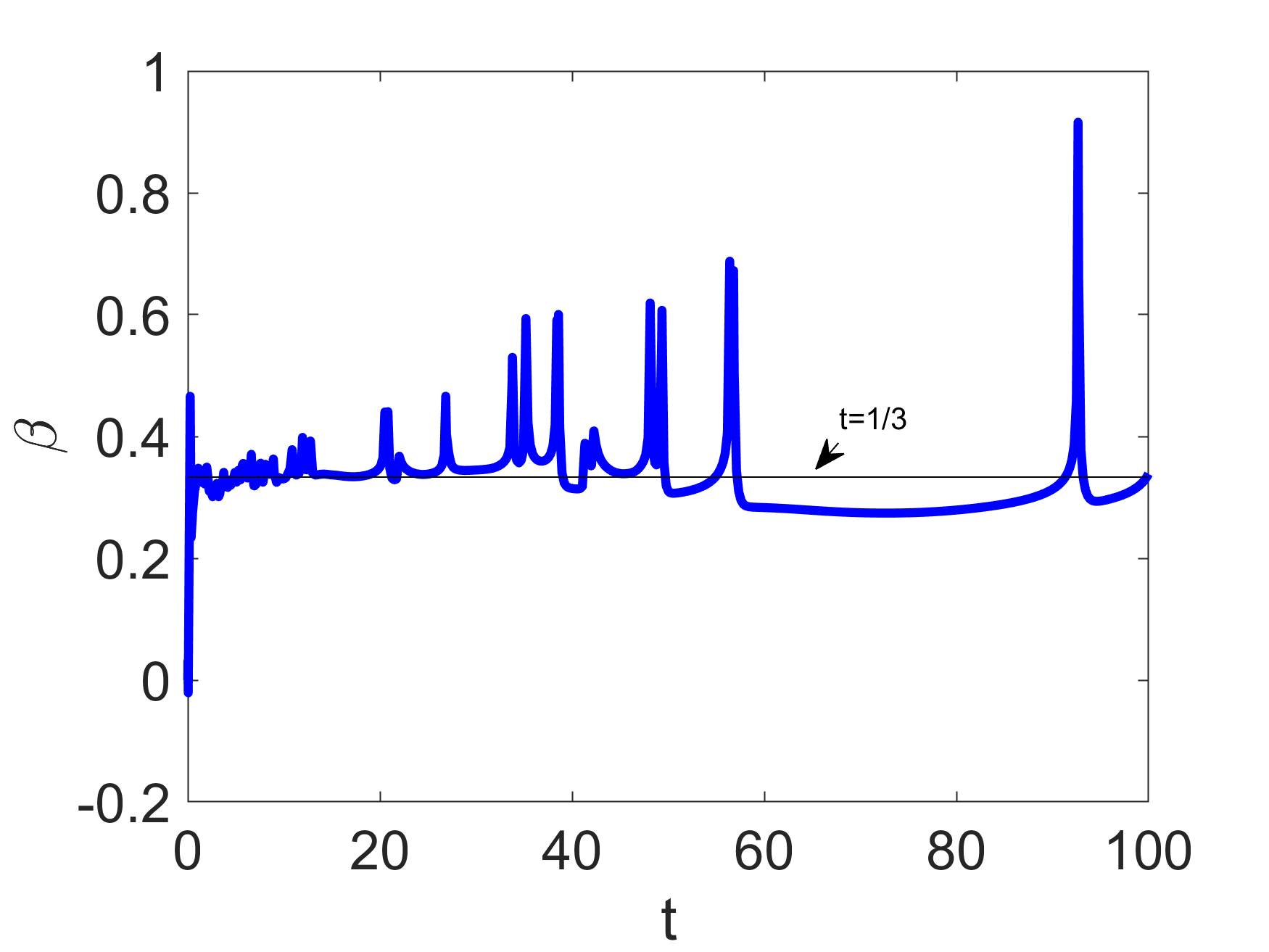}
		\caption{Time series of $\beta(t)$ for a single typical simulation of the Cahn--Hilliard equation (symmetric mixture, $L=16\pi$).  The line $\beta(t)=1/3$ is added  for reference.}
	\label{fig:2Dbeta_symmetric}
\end{figure}

We again perform a campaign of numerical simulations to characterize the coarsening rate $\beta$ in the Cahn--Hilliard equation.   An ensemble of simulations is constructed, in identical fashion to that already outlined in  Table~\ref{tab:ensemblesCH}.  In this way, the statistics of the coarsening rate $\beta$ are obtained.
\begin{figure}[htb]
	\centering
	\subfigure[\,\,\text{Mean}]{\includegraphics[width=0.45\textwidth]{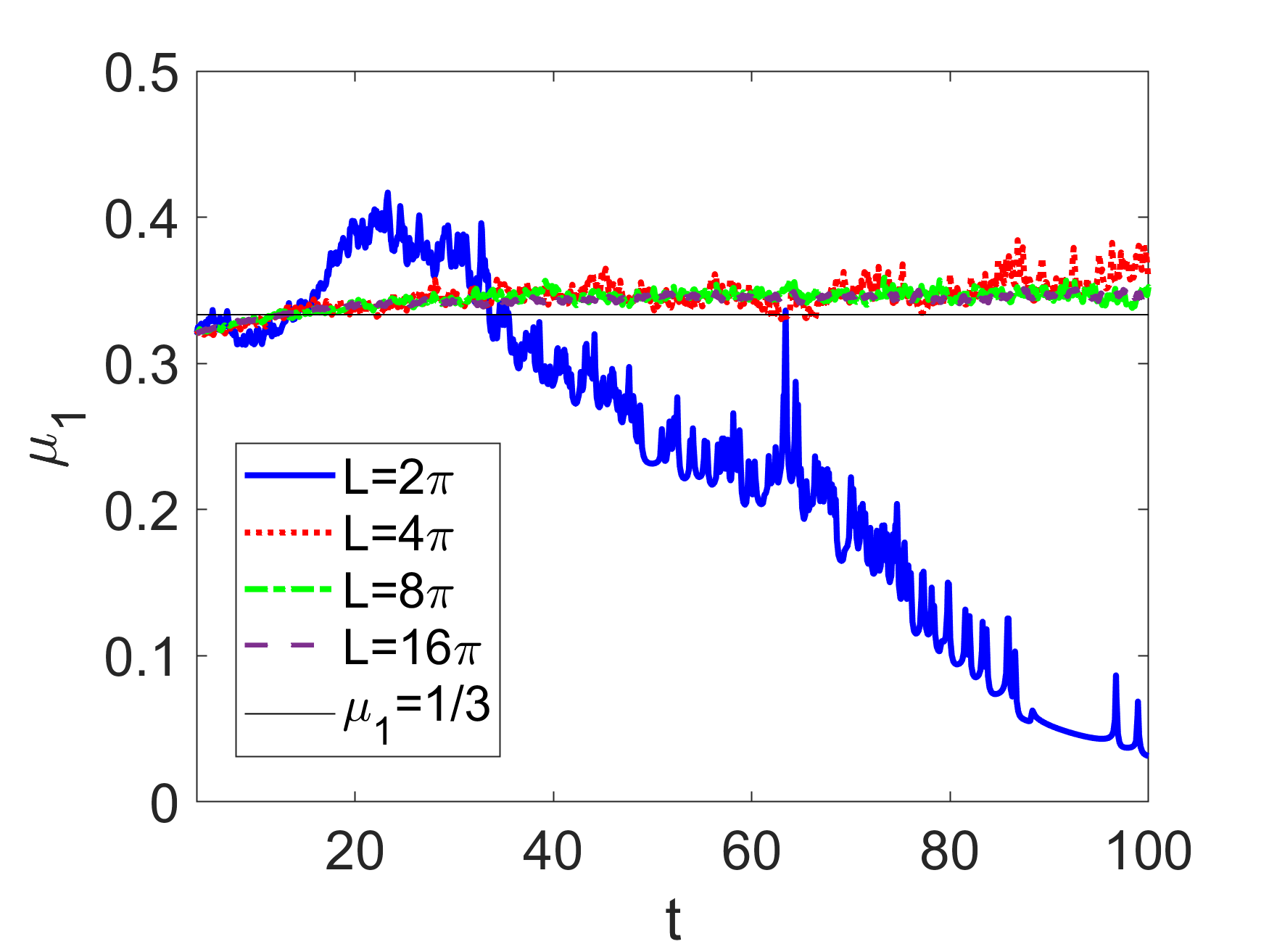}}
	\subfigure[\,\,\text{Variance}]{\includegraphics[width=0.45\textwidth]{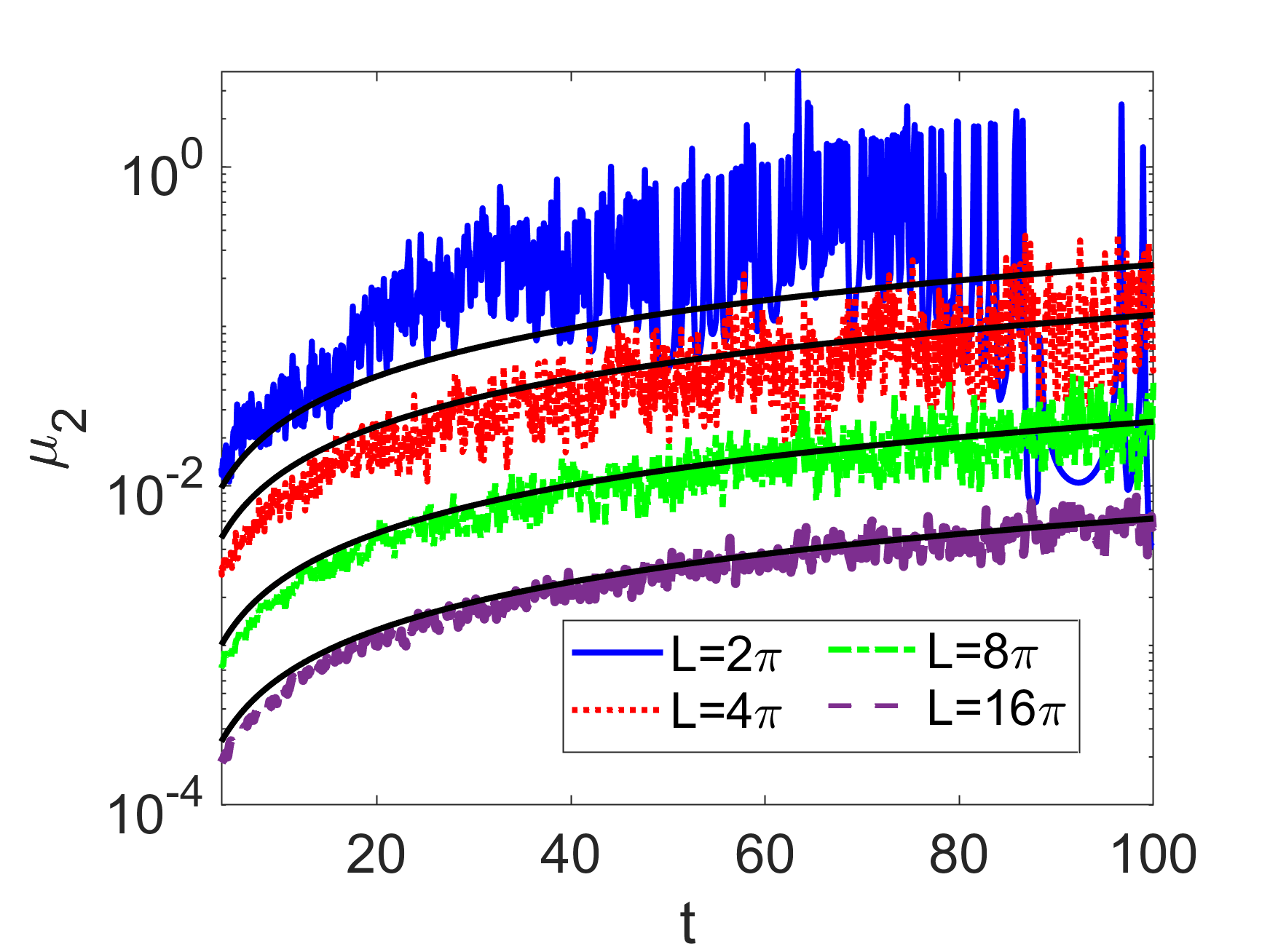}}
	\caption{Plots showing time series of  $\mu_1$ and $\mu_2$ (symmetric mixtures).  Moving averages are shown to guide the eye.  The plot in (b) uses a semilog scale to capture the extent of the variance in $\mu_2$, especially for the $L=2\pi$ case.    Least-squares fits are also shown in (b), with $\mu_2(t)=kt$.  This is to highlight the systematic drift in $\mu_2(t)$, and to guide the eye.}
	\label{fig:moments_CH_symmetric}
\end{figure}
The statistical moments are generated and the results  are shown in time-series form in Figure~\ref{fig:moments_CH_symmetric}.
Extreme finite-size effects are in evidence for the case $L=2\pi$ -- this case is discarded and is no longer considered here.
Otherwise, $\beta(t)$ is seen to fluctuate around a constant value $\beta_0$, and the amplitude of the fluctuations increases over time.  The amplitude of the fluctuations again seen to increase linearly in time, such that $\mathbb{E}(\delta\beta^2)\propto t$, in accordance to the asymmetric case.  The constant of proportionality is again estimated by least-squares fitting of $\mu_2(t)$ to a trend-line $\mu_2=kt$.  The results of applying this statistical model to the data are summarized in Table~\ref{tab:stats_CH_symmetric}.
\begin{table}
\centering
\begin{tabular}{|c|c|c|} 
\hline
    {$L$} & $\frac{1}{80}\int_{80}^{100}\mu_1(t)\,\mathd t$ & $k$ \\
    \hline
		\hline
    $2\pi$  & N/A       & N/A  	\\
    $4\pi$  & 0.356      & 0.0011   \\
    $8\pi$  & 0.346      & 0.00025 	\\
    $16\pi$ & 0.346      & 0.00062		\\
		\hline
\end{tabular}
\caption{Estimates of $\mathbb{E}(\beta)=\betaconst$ and $\mathbb{E}(\delta\beta^2)$ for Model 2 (symmetric mixtures), for various problem sizes $L$.  The value of $k$ is obtained from least-squares fitting on the data between $t=20$ and $t=100$.}
\label{tab:stats_CH_symmetric}
\end{table}
The time-averaged value of $\mu_1$ (a proxy for $\beta_0=\mathbb{E}(\beta)$) is seen to exceed $1/3$.  The robustness of this result to changes in grid refinement and the sampling window  is explored below.  Meanwhile, we emphasize that the variance parameter $k$  decreases as $k\sim L^{-2}$ for the considered cases in Table~\ref{tab:stats_CH_symmetric} -- i.e. the same trend as for asymmetric mixtures.
We lastly look at the probability distribution function of $\beta$ for symmetric mixtures and the Cahn--Hilliard equation --  this is done in Figure~\ref{fig:hist_symmetric}, for the case $L=16\pi$.
There are some similarities between this probability distribution function for symmetric mixtures and the corresponding distribution for asymmetric mixtures (Figure~\ref{fig:hist_asymmetric}).  Both have a well-defined peak and a broad positive tail, which increases in width over time.  

\begin{figure}[htb]
    \centering
     \includegraphics[width=0.6\textwidth]{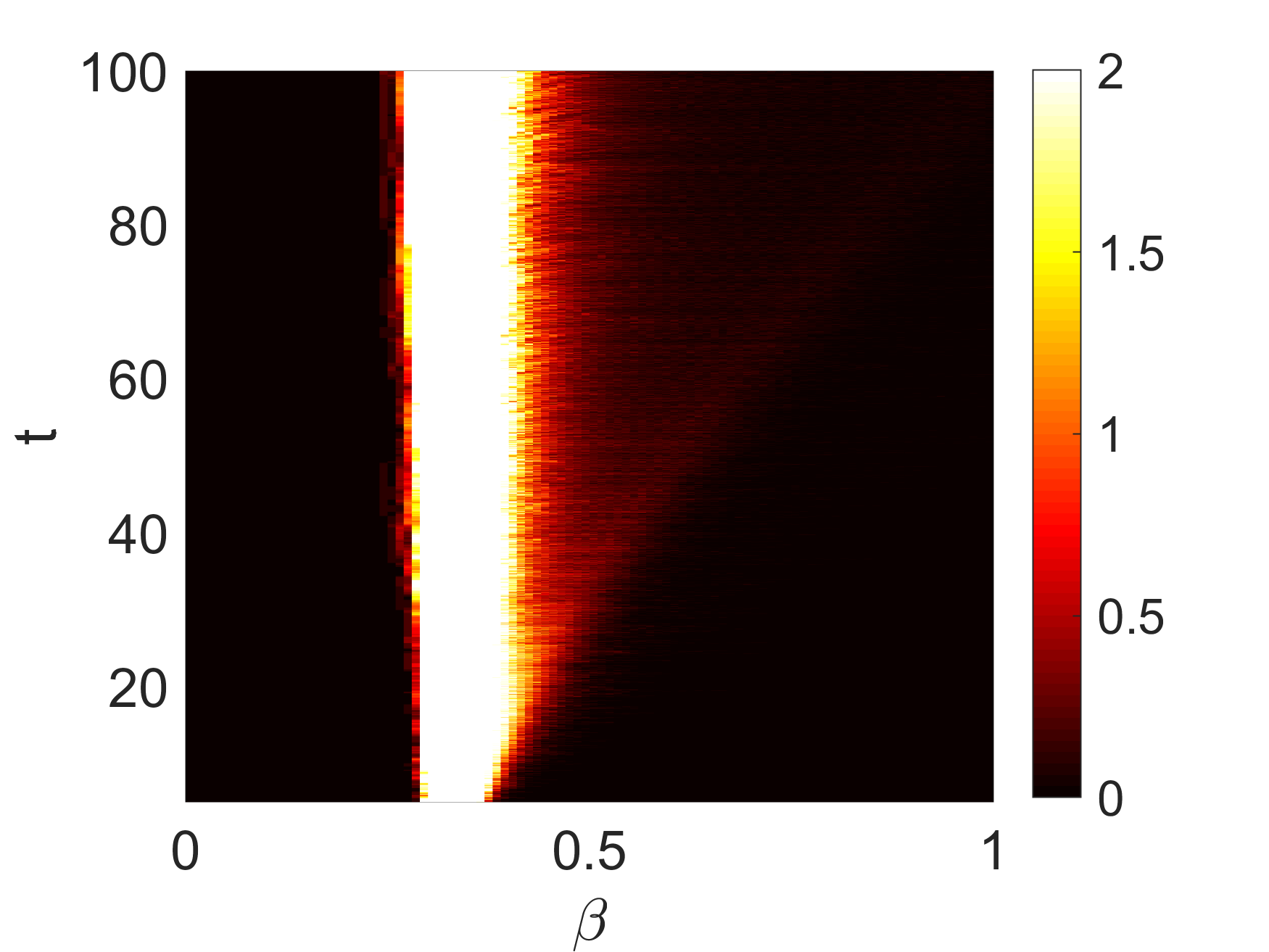}
				\caption{Spacetime plot of the histogram of $\beta$ for symmetric mixtures and the Cahn--Hilliard equation, for the case $L=16\pi$.     }
    \label{fig:hist_symmetric}
\end{figure}

\subsection{Discussion}

In this section we discuss in detail the measured values of the coarsening rates for both symmetric and asymmetric mixtures.    We start by looking at the robustness of the measured values of the coarsening rate for the case of the symmetric mixtures.
In order to assess the robustness of the numerical estimates of $\beta_0$ to variations in both the grid size and the sampling window, we tabulate the numerical values of $(t_2-t_1)^{-1}\int_{t_1}^{t_2}\mu_1(t)\mathd t$ for different time intervals $[t_1,t_2]$ and different grid sizes $n$ for the special case $L=4\pi$.  This is done in Figure~\ref{fig:refinement}.
%
%
%
%
%
%
There is little-or-no variation in the estimated values of $\beta_0$ as the sampling window and the spatial resolution $n$ are varied, with $n$ between $512$ and $1024$.  In particular, the estimated value of $\beta_0$ always exceeds $1/3$.  There is no systematic trend indicating $\beta_0\rightarrow 1/3$ as $n\rightarrow\infty$.  Therefore, the numerical evidence would suggest that the pointwise bound in Equation~\eqref{eq:kohn_maybe} does not hold.  However, since the estimated value of $\beta_0$ is close to $1/3$, the possibility that $\beta_0=1/3$ for very large systems and at very late times cannot be ruled out.  Therefore, it can still be said that the numerical evidence is just about consistent with the possibility that $\beta_0=1/3$ for very large systems at at very late times.

\begin{figure}[H]
	\centering
		\includegraphics[width=0.9\textwidth]{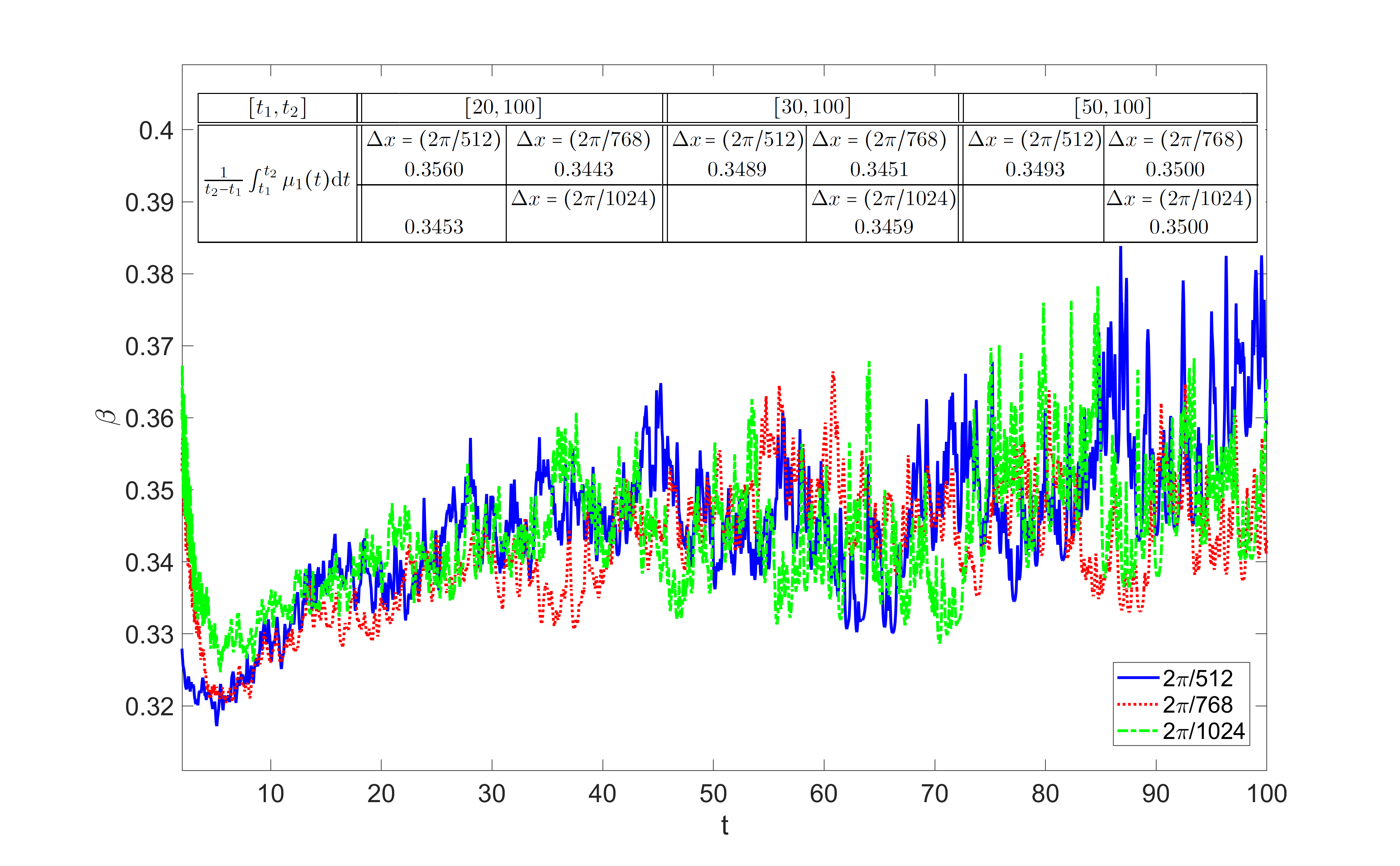}
		\caption{Sensitivity analysis for a fixed domain size $L=4\pi$.  The table shows the robustness of the estimated value of $\beta_0$ to changes in the grid resolution $\Delta x$ and the sampling window $[t_1,t_2]$. }
	\label{fig:refinement}
\end{figure}

We also compare the measured coarsening rates for asymmetric mixtures ($\beta_0\approx 0.28$), and for symmetric mixtures ($\beta_0\approx 0.35$).  There is a significant difference between these values.  Indeed, as the asymmetric mixture is a close approximation to the idealized model of Ostwald Ripening considered in Section~\ref{sec:drop_pop}, one would expect the coarsening rate to be closer to $1/3$ for the asymmetric mixture.  This discrepancy may be due to the fact that simulation time (while extending to $t=100$) may not be long enough  to capture true late-stage coarsening.  Previous numerical simulations looking at the onset of late-stage coarsening have shown that the scaling behaviour (where the typical droplet size grows unambiguously as a power law) is considerably delayed (by orders of magnitude) for the asymmetric case, compared to the symmetric case (e.g. Reference~\cite{garcke2003transient}).   Similarly, longer simulation times for the symmetric case may reduce the measured values of $\beta_0$, thereby possibly bringing them down and closer to $1/3$.

It should be emphasized however that in simulations, one is constantly constrained by two competing effects: the need to extend the simulations out to late times to capture the power-law behaviour of the coarsening, and the need at the same time to avoid finite-size effects which appear at precisely such late times.  The simulation times chosen in this work reflect an attempt to balance these competing effects, and in any case produce results that are consistent with the pointwise bound~\eqref{eq:kohn_maybe}.

\section{Conclusions}
\label{sec:conc}

In this work, we have looked at the coarsening rate of droplets undergoing Ostwald Ripening in two contexts: a discrete droplet-poulation model (`Model 1'), and the Cahn--Hilliard model (`Model 1').  The classical LSW theory emerges from Model 1 in the limit as $N$, the number of droplets initially present, tends to infinity.
 We have quantified the coarsening rate by reference to $\beta$, the growth rate $\beta=-(t/F)(\mathd F/\mathd t)$, where $F$ is the free energy of the droplet system.  We have addressed the following questions:

\begin{enumerate}
\item Is it sensible even to define a probability distribution function for $\beta$? 

In LSW theory we have demonstrated that this is sensible: the probability distribution function for $\beta$ (a property of the entire system and not just of a single droplet) is sharp and equal to $\delta(\beta-(1/3))$.  
\end{enumerate}

\noindent This then justifies the formulation of an analogous probability distribution function for Model 1, and hence:

\begin{enumerate}[resume]
\item  What is  the probability distribution function for $\beta$ in Model 1?  How do finite-size effects (parametrized by $N$, the number of droplets initially present) alter the shape of the distribution?  

In this case, the probability distribution function is no longer sharp, or indeed, statistically stationary.  Here, the probability distribution function is constructed from an ensemble of numerical simulations; the evidence from this ensemble  is that the probability distribution function in this case is broad, and that the variance increases systematically over time.  However, the variance also decreases linearly with increases in $N$, the number of droplets initially present in a given simulation.

\end{enumerate}

\noindent This has motivated us to apply the same approach to Model 2, and hence:

\begin{enumerate}[resume]
\item  What is the probability distribution function for $\beta$ in Model 2?  How do finite-size effects (parametrized by $|\Omega|$) alter the shape of the distribution?  

Here, we again construct the probability distribution function of $\beta$ from an ensemble of numerical simulations: the evidence again shows that the probability distribution function for $\beta$ is not sharp or statistically stationary.  The probability distribution function  is broad, and the variance increases systematically over time.  However, the variance also decreases linearly with increases in $|\Omega|=L^2$, the size of the computational domain.
\end{enumerate}

In each of the models considered, the numerical evidence is further consistent with a coarsening rate that satisfies Equation~\eqref{eq:kohn_maybe}, i.e. $L(t)\leq (\mathrm{Const.})t^{1/3}$, where $L(t)=|\Omega|/[\text{Total Interfacial Area at time }t]$; the total interfacial area may in turn be identified (up to a prefactor) with the free energy of the droplet system.   The evidence for this bound for the case of Model 2 and symmetric mixtures can be queried: the simulations are hampered by the need to run simulations for very long times to capture the true late-stage coarsening behaviour and by the competing need to stop simulations before extreme finite-size effects become a problem.  The order of accuracy of the numerical method for the Cahn--Hilliard equation could also be improved (e.g. Appendix~\ref{sec:app:cvg}).  However, the present approach establishes a framework for further numerical investigations of the existence or otherwise of the pointwise bound $L(t)\leq (\mathrm{Const.})t^{1/3}$.

\subsection*{Acknowledgements}

This work has been produced as part of ongoing work within the ThermaSMART network. The ThermaSMART network has received
funding from the European Union’s Horizon 2020 research and innovation programme under the Marie Sklodowska--Curie grant
agreement No. 778104.  AG and L\'ON  gratefully acknowledge the support of NVIDIA Corporation with the donation of the Titan X Pascal GPUs used for this research.   AG acknowledges funding received from the UCD Research Demonstratorship.  
L\'ON acknowledges helpful discussions with Neil O'Connell.

\appendix

\section{Convergence of the Finite-Difference Method}
\label{sec:app:cvg}

In this Appendix we look at the convergence of the finite-difference code used to simulate the Cahn--Hilliard equation in Section~\ref{sec:CH}.  The code is the based on a semi-implicit Alternating Direction Implicit (ADI) finite-difference method, with periodic boundary conditions in both spatial directions.    The computational grid has $n$ gridpoints in each spatial direction, and the time-step is denoted by $\Delta t$.
The size of the physical domain is $L$, such that $|\Omega|=L^2$.
The numerical method is described in detail in References~\cite{gloster2019cupentbatch,gloster2019custen} -- we present a further convergence study here for completeness.

Since analytical solutions of the Cahn--Hilliard equation in two dimensions are difficult to come by, we resort to a numerical benchmark.  The numerical benchmark methodology adopted herein is that presented in Reference~\cite{ding2007diffuse}.  This benchmark is performed by successively refining grids with the same initial condition and then comparing solutions at matching points.  Matching points is easily achieved by repeatedly doubling the total number of points in the domain.  Thus, the quantity we  are looking to compute is given by:
\begin{multline}
E_n=\frac{1}{|\Omega|}\int_{\Omega}|C_n(\vecx,t)-C_{n/2}(\vecx,t)|\mathd^2x\\
\approx
\frac{1}{|\Omega|}\sum_{i=1}^{n/2}\sum_{j=1}^{n/2}\big| 
\tfrac{1}{4}\left[C_n(\vecx_{2i-1,2j-1},t)+
C_n(\vecx_{2i-1,2j},t)+
C_n(\vecx_{2i,2j-1},t)+
C_n(\vecx_{2i,2j},t)\right]\\
-
C_{n/2}(\vecx_{i,j},t)\big|\Delta x_{n/2}^2,
\end{multline}
where $C_n$ denotes the discretized Cahn--Hilliard concentration field on an $n\times n$ grid and $C_{n/2}$ denotes the same concentration field on an $(n/2)\times (n/2)$ grid.  Here, we assuming that there is an equal grid spacing in both spatial directions.  Hence $\Delta x_n$ denotes the uniform grid spacing on the $n\times n$ grid and 
$\Delta x_{n/2}$ denotes the uniform grid spacing on the $(n/2)\times (n/2)$ grid.  Hence finally, $\Delta x_{n/2}=2\Delta x_n$

For the purpose of the convergence study we work on a domain $\Omega=[0,2\pi]^2$; the initial condition is set in this context as
\[
C(\vecx,t=0)=\epsilon\tanh(r-\pi),\qquad \vecx=(x,y),\qquad r=\sqrt{x^2+y^2}.
\]
In this way, we simulate the Ostwald ripening of a single droplet.
We take $\epsilon=10^{-6}$, and $\gamma=0.01$.  The final simulation time is $T=10$.  We choose a time step of $\Delta t=0.1\Delta x$, where $\Delta x$ is the uniform grid spacing in the $x$ and $y$-direction for a particular number of gridpoints $n$.  

We can see from the results presented in Figure~\ref{fig:cvg} and Table~\ref{tab:cvg} that the considered numerical scheme does indeed converge and is a first-order accurate scheme.
\begin{figure}[htb]
	\centering
		\includegraphics[width=0.6\textwidth]{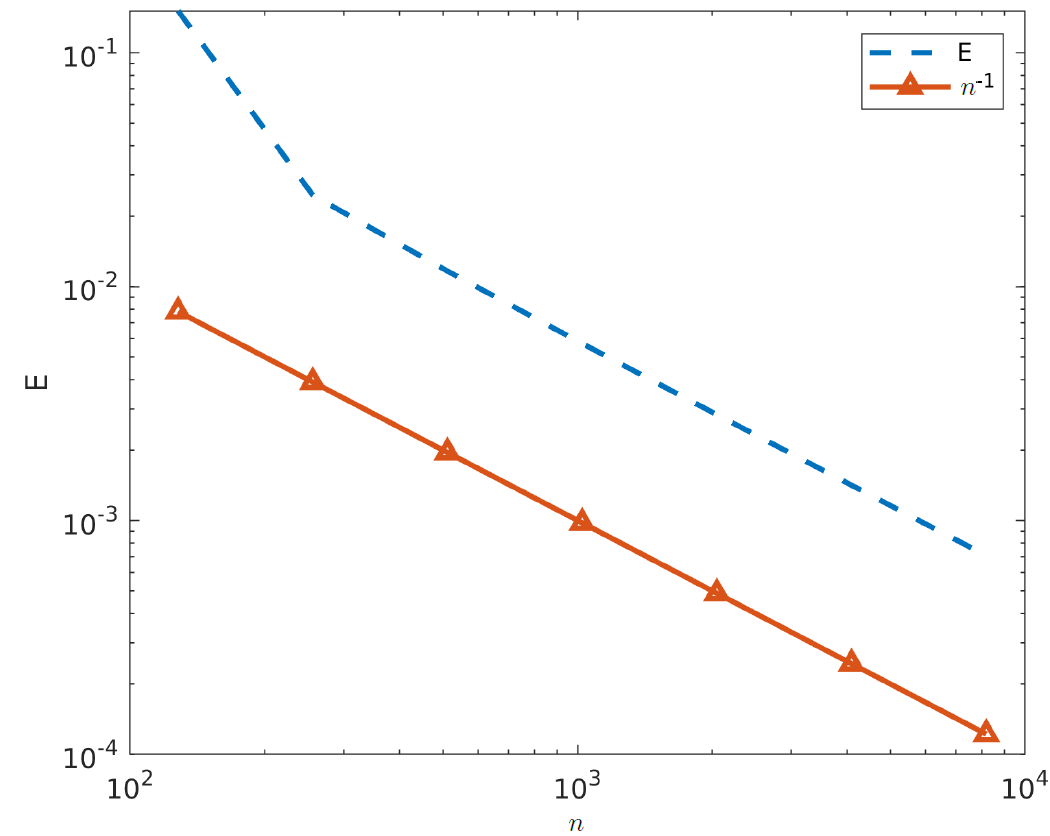}
		\caption{Plot showing the convergence of the ADI numerical scheme for the Cahn--Hilliard equation}
	\label{fig:cvg}
\end{figure}
Another version of the scheme which solves the hyper-diffusion equation has been shown to be second-order accurate, but the addition of the non-linear Cahn--Hilliard term reduces toe scheme to first-order accuracy.
\begin{table}
	\centering
		\begin{tabular}{c|c|c} 
		$n$  & $E_n$  & $\mathrm{log}_2(E_n/E_{2n})$ \\
		\hline
		128  & 0.1510 & 2.6197  \\
		256  & 0.0246 & 1.0785  \\
		512  & 0.0116 & 1.0281  \\
		1024 & 0.0057 & 1.0107 \\
		2048 & 0.0028 & 1.0022 \\
		4096 & 0.0014 & 0.9996\\
		9192 & 0.0007 &       \\
		\end{tabular}
		\caption{Table showing rates of convergence for simulation of the 2D Cahn--Hilliard equation using the ADI numerical method}
		\label{tab:cvg}
\end{table}

\section*{References}

\bibliographystyle{unsrt}

\end{document}